\documentclass [11pt]{article}
\usepackage{amsmath,amsthm,amsfonts,amscd,eucal,latexsym,amssymb,mathrsfs,tikz,color,cancel}
\usetikzlibrary{decorations.markings,shapes.misc}
\usepackage{epsfig}  % ciao
\usepackage{simplewick}
\usepackage{hyperref}
\input xy
\xyoption{all}
 
\tikzset{
GFleche/.style={
 draw=black,
 postaction={decorate},
 decoration={markings,mark=at position 0.5 with {\arrow{<}}}
 },
DFleche/.style={
 draw=black,
 postaction={decorate},
 decoration={markings,mark=at position 0.5 with {\arrow{>}}}
 }}

\definecolor{NiColor}{RGB}{77,77,255}

\def\WF{{\text{WF}}}

\def\beq{\begin{eqnarray}}
\def\eeq{\end{eqnarray}}

\def\supp{\textrm{supp}}
\def\supp{\textrm{supp}}

\newcommand{\dvol}[1]{\textrm{d}#1}

\newcommand{\vettore}[1]{{\bf #1}}

\def\remark {\vsp\ifhmode{\par}\fi\noindent\noindent {\bf Remark:} 
}

\DeclareMathOperator*{\vlim}{v-lim}

\newtheorem{theorem}{Theorem}[section]

\newtheorem{proposition}[theorem]{Proposition}

\newtheorem{definition}[theorem]{Definition}

\begin{document} 
% 
%\hfill{\sl Preprint UTM XYZ - February 2013} 
%\par 
%\bigskip 
%\par 
%\rm 
% 
%%%%%%%%%%%%%   Title %%%%%%%%%%%%%%%%%%%%%%%%%% 
 
\par 
\bigskip 
\LARGE 
\noindent 
{\bf Relative entropy and entropy production for equilibrium states in pAQFT} 
\bigskip 
\par 
\rm 
\normalsize 
 
%%%%%%%%%%%%%%%%%%%%%%%%%%%%%%%%%%%%%%%%%%%%%
%%%%%%%%%%%% Authors %%%%%%%%%%%%%%%%%%%%%%%% 

\large
\noindent 
{\bf Nicol\`o Drago$^{1,2,a}$}, {\bf Federico Faldino$^{3,4,b}$}, {\bf Nicola Pinamonti$^{3,4,c}$} \\
\par
\small
\noindent$^1$ Dipartimento di Fisica, Universit\`a di Pavia - Via Bassi, 6, I-27100 Pava, Italy. \smallskip

\noindent$^2$ Istituto Nazionale di Fisica Nucleare - Sezione di Pavia, Via Bassi, 6, I-27100 Pava, Italy. \smallskip

\noindent$^3$ Dipartimento di Matematica, Universit\`a di Genova - Via Dodecaneso, 35, I-16146 Genova, Italy. \smallskip

\noindent$^4$ Istituto Nazionale di Fisica Nucleare - Sezione di Genova, Via Dodecaneso, 33 I-16146 Genova, Italy. \smallskip
\smallskip

\noindent E-mail: 
$^a$nicolo.drago@unipv.it, 
$^b$faldino@dima.unige.it,
$^c$pinamont@dima.unige.it\\ 

\normalsize

\par 
 
\rm\normalsize 

%\linespread{1.5} 
\rm\normalsize 
 
%%%%%%%%%%%% Date %%%%%%%%%%%%%%%%%%%%%%%%%%  
%%%%%%%%%%%% Date %%%%%%%%%%%%%%%%%%%%%%%%%% 
 
\par 
\bigskip 
 
\rm\normalsize 
%\noindent {\small Version of \today}

\par 
\bigskip

%\linespread{1.5} 
\rm\normalsize

\bigskip

\noindent 
\small 
{\bf Abstract}.
We analyze the relative entropy of certain KMS states for scalar self-interacting quantum field theories over Minkowski backgrounds that have been recently constructed by Fredenhagen and Lindner in \cite{FredenhagenLindner} in the framework of perturbative algebraic quantum field theory. 
The definition we are using is a generalization of the Araki relative entropy to the case of field theories.  
In particular, we shall see that the analyzed relative entropy is positive in the sense of perturbation theory, hence, even if the relative modular operator is not at disposal in this context, the proposed extension is compatible with perturbation theory. 
In the second part of the paper we analyze the adiabatic limits of these states showing that also the density of relative entropy obtained dividing the relative entropy by the spatial volume of the region where interaction takes place is positive and finite.
In the last part of the paper we discuss the entropy production for states obtained by an ergodic mean (time average) of perturbed KMS states evolved with the free evolution recently constructed by the authors of the present paper. We show that their entropy production vanishes even if return to equilibrium \cite{Robinson, HKT} does not hold. This means that states constructed in this way are thermodynamically simple, namely they are not so far from equilibrium states.
%setups.

\normalsize

\vskip .3cm

\section{Introduction}

Recently, Lindner in is PhD Thesis \cite{Lindner} and Fredenhagen and Lindner in \cite{FredenhagenLindner} have constructed equilibrium states in the adiabatic limit for a perturbatively constructed interacting quantum scalar field which propagates on a Minkowski spacetime. In this paper we show how relative entropy among equilibrium states for different interactions can be given within perturbation theory. 

The analysis performed in \cite{Lindner, FredenhagenLindner} 
is done in the context of perturbative algebraic quantum field theory (pAQFT) \cite{BDF, FredenhagenRejzner, FredenhagenRejzner2, HW01, HW02, HW03}.
According to the algebraic paradigm, the emphasis is posed on the set of observables of the theory and on the relations among them. In particular, the observables form a $*-$algebra that we indicate by $\mathcal{A}(M)$ where $M$ is the Minkowski spacetime\footnote{Many of the concepts presented here can be generalized to the case of globally hyperbolic spacetimes \cite{BaGi12}.} where fields propagate. A review of the construction of the $*-$algebra of the free scalar Klein Gordon field can be found e.g. in \cite{BFV}. 
The self-interaction of a scalar field $\phi$ is described by a Lagrangian density of the form
\[
\mathcal{L}  =  \partial_\mu \phi\partial^\mu \phi +m^2 \phi^2 + \lambda \mathcal{L}_I(\phi) 
\]
where $\mathcal{L}_I$ is the interaction Lagrangian density, which is assumed to be a formally selfadjoint constant section of the jet bundle constructed over $\phi$ which is not simply quadratic in the field. The corresponding non linearities in the field equations can be treated with perturbation theory.
According to pAQFT the observables of the interacting algebra $\mathcal{A}_I(M)$ can be represented as elements of $\mathcal{A}(M)[[\lambda]]$, the algebra of formal power series in the coupling constant $\lambda$ with coefficients in the free algebra ($\lambda=0$). This representation is realized  by the quantum M\o ller map, whose action is given by the Bogoliubov formula, 
\begin{equation}\label{eq:bogoliubov}
R_V(F)=S^{-1}\star (S \cdot_T F)\,,
\end{equation}
where $F$ is a local interacting field.
Furthermore $V$ is the smeared interaction Lagrangian $\mathcal{L}_I$ and $S$ is the time ordered exponential of the smeared interaction Lagrangian $V$.
Finally $\star$ denotes the product in $\mathcal{A}$ and $\cdot_T$ is the corresponding time ordered product. The algebra of both free and interacting theories can be analysed in a state independent way and this task is nowadays well understood. 

A state of the theory is described by a linear positive normalized functional over $\mathcal{A}(M)$. Once a state is chosen, the standard representation of any quantum theory in terms of operators on a Hilbert space can be recovered applying the celebrated Gelfand Neimark Segal (GNS) construction. The analysis of the state space in perturbation theory is not straightforward, in particular, having a state of the free algebra, we can obtain a state for the interacting algebra simply composing with the M\o ller map.  However, the physical meaning of the state is modified in this procedure. In particular, if one starts with an equilibrium state (a state which satisfies the Kubo Martin Schwinger (KMS) condition \cite{HHW}) over $\mathcal{A}(M)$ the composition with the M\o ller operator  leads to a state which does not satisfy the equilibrium property anymore. 

In \cite{Lindner, FredenhagenLindner} the authors managed to find the way to modify an equilibrium state for the free theory $\omega^\beta$ so to obtain an equilibrium state for the interacting theory.
This construction is an extension of the work of Araki \cite{Araki-KMS}, 
about the perturbation of KMS states in the context of von Neumann theories, 
see also \cite{BR}, to pAQFT. Furthermore, the formulation of equilibrium states proposed in those papers survives the adiabatic limits, namely the limit where the coupling constant tends to one, with the caveat that all the elements of the theory are understood in the sense of formal power series.  

In the literature, equilibrium states for interacting fields are constructed by means of Keldysh formalism 
see e.g. \cite{Landsman, LeBellac} and reference therein.  However, infrared problems at higher perturbative orders are not completely avoided in that manner. 
The imaginary time formalism is more efficient in this respect, in particular the perturbative construction of the partition function of the state can usually be performed with the imaginary time formalism, expanding the propagators over the Matsubara frequencies.
Unfortunately, in this formalism, the direct computation of the correlation in position space require a backward Wick rotation which is not completely under control.
Actually, a complete construction which extends the work of Araki was possible only thanks to some recent achievements in pAQFT.

In particular, the first new ingredient is the time slice axiom, by means of which, being $\mathcal{A}(M)$ the algebra of observables, a state over $\mathcal{A}(M)$ is fixed once it is known on the subalgebra $\mathcal{A}(\Sigma_\epsilon)$ of observables supported on a small neighborhood $\Sigma_\epsilon\simeq\Sigma\times(-\epsilon,\epsilon)$ of a Cauchy surface $\Sigma$ of $M$. 
This fact holds for both free and perturbatively constructed interacting quantum field theories \cite{CF}.

Second, the causal factorization properties of the $S-$matrix, used to construct the M\o ller map 
implies that algebras of interacting fields $\mathcal{A}_I(\Sigma_\epsilon)$ constructed perturbatively with interaction Lagrangians which coincides over 
$\Sigma_\epsilon$ are equivalent up to isomorphisms. 

For this reason,
in \cite{Lindner} and \cite{FredenhagenLindner}, the authors
considered interaction Lagrangians of the form 
\begin{equation}\label{eq:perturbation-lagrangian}
V_\chi^h = \int  \chi(t) h({\bf x}) \mathcal{L}_I(t,{\bf x})\;  d t\, d^3 {\bf x}\,,
\end{equation}
where $(t,{\bf x})$ are standard Minkowski coordinates used to parametrize the points of the Minkowski space $M$, $\chi\geq 0$ is a function of time which makes the support of $V_\chi^h$ past compact and it is constructed in such a way that $\chi=1$ on $J^+(\Sigma_\epsilon)$, $h\geq 0$ is a spatial cutoff. Furthermore, the interaction Lagrangian density $\mathcal{L}_I$ is assumed to be invariant under translations. In the subsequent part of the paper, we shall simply denote the smeared interaction Lagrangian with $V$ whenever the dependence on $\chi$ and $h$ is not strictly necessary for the understanding.
In this way, $V_{\chi}^{h}$ is a local field and thus the perturbative construction of interacting fields by means of the M\o ller map \eqref{eq:bogoliubov}, 
is both ultraviolet and infrared finite. 
The algebra of interacting fields in the adiabatic limit can be easily constructed by means of the observations given above considering the algebra $\mathcal{A}_I(\Sigma_\epsilon)$, generated by $R_V(\mathcal{F}_{\text{loc}})$, where $\mathcal{F}_{\text{loc}}$ are interacting local fields, and taking the inductive limit $h\to 1$ while keeping $\chi$ fixed \cite{BF00}.
We stress once more that although this procedure realizes the interacting algebra $\mathcal{A}_I(\Sigma_\epsilon)$ as a $*$-subalgebra of $\mathcal{A}(M)[[\lambda]]$, the $*$-algebra of formal power series in the coupling $\lambda$ with coefficients in $\mathcal{A}(M)$,
the algebras $\mathcal{A}_I(\Sigma_\epsilon)$ and $\mathcal{A}(M)[[\lambda]]$ are $*$-isomorphic.
As mentioned above this is again due to the time-slice axiom and to the fact that $R_V(F)=F$ if the support of $F$ lies in the past of the support of $V$.
In particular, this implies that states on $\mathcal{A}(M)$ can be regarded as states on $\mathcal{A}_I(\Sigma_\epsilon)$ and viceversa.

The perturbative construction of KMS states given by Araki in \cite{Araki-KMS} can be repeated for $\mathcal{A}_I(\Sigma_\epsilon)$ comparing the automorphisms of both the free and interacting time evolution. However, these states depend on the spatial cutoff $h$, in \cite{FredenhagenLindner} it is shown that the limit $h\to 1$ of the expectation values of interacting fields can be taken when the background free theory is a massive Klein Gordon field theory. The case of certain massless background theories as for the case $\mathcal{L}_I=\phi^4$ can be addressed adding the thermal mass present in that context to the background theories by means of the principle of pertrubative agreement, see \cite{HW05, DHP}. 

Standard results, which hold in the case of Quantum Statistical Mechanics, can be generalized to this case. 
As an example, the return to equilibrium property given in the context of $C^*-$dynamical systems in \cite{BrKiRo,Robinson} has been recently extended to the case of field theories in \cite{DFP}.

\bigskip

In this paper we address the relative entropy \cite{Araki, Araki2} and entropy production \cite{Oj0,Oj1, Oj2, JP01, JP02} among the above mentioned class of states.
In particular, we will exploit the time-slice axiom in order to identify $\mathcal{A}_I(\Sigma_\epsilon)$ as $\mathcal{A}(M)[[\lambda]]$, so that all states mentioned above can be compared within the same algebra.
Within this setting, we shall see that the relative entropy can be computed and it gives finite results.
Thus, the relative entropy is a concept compatible with perturbation theory.
Furthermore, we shall see that the relative entropy is positive in the sense of perturbation theory. This shows that at least in this case the relative entropy can be expressed by means of fields, contrary to the general case where the relative entropy is given in terms of the relative modular operator which is not available in the algebra of interacting field observables. 

More precisely, the formula for the relative entropy we shall give below \eqref{eq:relative-entropy-KMS-generalized} can be tested among copies of equilibrium states $\omega^{\beta, V_1}, \omega^{\beta, V_2}$ constructed with different interaction Lagrangians and, when $V_1=0$ and $V_2=V$, it reduces to the \[
\mathcal{S}(\omega^\beta, \omega^{\beta,V}) =    \omega^\beta(\beta K) +\log (\omega^\beta(U(i\beta)))\,,
\]
where $\log \omega^\beta(U(i\beta)))/\beta$ is the relative free energy and $K$, which is the generator of the cocycle which intertwines the free and interacting time evolution in $\mathcal{A}(M)[[\lambda]]$, plays the role of the perturbation of the Hamiltonian of the system. 
The formula is obtained by analogy with the Araki formula evaluated for different Araki states (see the appendix \ref{ap:rel-entropy}). We shall thus prove that the basic properties of the relative entropy are satisfied by it.

The formula for the relative entropy we give below   
in \eqref{eq:relative-entropy-KMS-generalized}
survives the thermodynamical limit ($h\to1$) if spacelike densities are considered, however, it can only be tested among the perturbatively constructed KMS states described above. 
In particular, evolving a perturbatively constructed KMS state with the free evolution by a time step $t$ and taking the limit where $t$ tends to infinity, one gets back the KMS state for the free theory when the interaction Lagrangian has compact spatial support. On the contrary, if the adiabatic limit is considered this does not hold anymore. Actually, in \cite{DFP} it was shown that the ergodic mean, namely the time average, of this states converges to a non equilibrium steady state (NESS). We would like to evaluate how far this NESS is from one of the perturbatively constructed KMS states and usually the (symmetric part of the) relative entropy could be used to reach this goal. Unfortunately, we cannot directly apply the given definition. To this end, in the second part of paper, we analyze the entropy production in these states following ideas present in \cite{Oj0,Oj1, Oj2, JP01, JP02}. We prove that the entropy production of the NESS described above vanishes. This implies that this state is not so far from an equilibrium state.

Indeed, the interest in non-equilibrium states and entropy production has increased in the recent past years both in Statistical Mechanics (see for instance \cite{Spohn Lebowitz}, where non-equilibrium steady states for a cristalline system is studied, or \cite{Spohn} for a definition of entropy production in the framework of quantum semigroups) and in Quantum Field Theory, in particular with reference to 2-dimensional Conformal Field Theory \cite{Doyon 2, Doyon 3} and, more recently \cite{Hollands Longo}, but also for the Klein-Gordon field \cite{Doyon 1}. 

\bigskip

The paper is organized as follows. First of all, in the second part of the introduction, we shall give a brief introduction to pAQFT and we recall the notion of relative entropy given by Araki.
In the second section we analyze the concept of relative entropy for KMS states of interacting theories for spatially compact perturbations and in the third section we discuss their adiabatic limit. Section 4 contains the discussion of the entropy production and of its use to evaluate how far from equilibrium  are certain non equilibrium steady states constructed by a time average of interacting KMS states evolved with the free evolution.
Some conclusion are presented in the fifth section. Finally, some minor or known useful results are collected in the appendix.

\subsection{Brief introduction to pAQFT}

Perturbative algebraic quantum field theory (pAQFT) is a recently developed framework in which ideas of algebraic quantum field theory \cite{Haag} are combined with methods proper of perturbation theory \cite{EG, Steinmann, Bogoliubov} to treat interacting quantum fields. 
The first paper where this formulation was presented is \cite{BDF}, and further developed in \cite{BF09, FredenhagenRejzner,FredenhagenRejzner2}. 
The analysis of interacting quantum field theory in the algebraic framework was previously formulated in 
\cite{BF00}, while the analysis of the curved case was presented in \cite{Kay Wald}, \cite{HW01, HW02, HW03}, \cite{BFV}.

\bigskip

In this paper we are interested in interacting scalar quantum field theory constructed by means of perturbation theory over free massive theories in a Minkowski spacetime $M$ whose metric $\eta$ is assumed to have signature $(-,+,+,+)$. The classical equation of motion of the theory we want to treat is the following 
\begin{equation}\label{eq:interacting-equation}
-\Box \phi + m^2 \phi +\lambda V^{(1)}(\phi) = 0\,,
\end{equation}
where $\Box=\nabla_\mu \nabla^\mu$ is the D'Alembert operator, $m$ is the mass of the field, $\lambda$ is a coupling constant which is often set to $1$ and $V^{(1)}(\phi)$ is a local interaction Lagrangian, it is actually the first functional derivative of \eqref{eq:perturbation-lagrangian} in the limit where both $\chi$ and $h$ are taken to be equal to $1$.
We briefly recall how this theory is treated in pAQFT.

The starting point is the choice of the space $\mathcal{C}$ of off-shell {\bf field configurations} which are assumed to be real-valued smooth functions over $M$, in particular, we shall denote the generic configuration as follows
\[
\phi\in\mathcal{C}:= {C}^{\infty}(M;\mathbb{R})\,.
\]
The set of {\bf observables} we are looking for are the (non linear) functionals over the field configurations. We shall restrict our attention to the functionals which admit functional derivatives to all order, which have functional derivatives that are compactly supported distributions and whose wave front set is microcausal. More precisely, the set of {\bf microcausal functionals} is denoted by 
\[
\mathcal{F}_{\mu c} := \left\{ F:\mathcal{C}\to \mathbb{C}  \left|  F^{(n)}\in \mathcal{E}(M^{n}), \,\WF (F^{(n)}) \cap ({\overline{V}^+}^{n}\cup {\overline{V}^-}^{n}  ) = \emptyset    , \,  \forall n\right. \right\}
\]
where $(x_1,\dots, x_n;p_1,\dots, p_n)\in {\overline{V}^\pm}^{n}\subset T^*{M^n}$ if, for every $i$, $p_i$ is a causal covector respectively future/past directed. Sometimes it is requested that the elements of $\mathcal{F}_{\mu c}$ have only finitely many nonvanishing functional derivatives. If this is the case, only polynomial interactions can be treated with the formalism. If more complicated interactions are treated, it is necessary to introduce some notion of convergence for having a well-defined quantum non commutative product, see \textit{e.g.} \cite{BDF}.
When equipped with the pointwise product and with the complex conjugation as involution, $\mathcal{F}_{\mu c}$ forms the commutative $*-$algebra of classical field observables. Relevant subsets of this algebra are the set of {\bf local functionals} 
\[
\mathcal{F}_{\text{loc}} := \left\{ F\in \mathcal{F}_{\mu c} \left|\,  \supp F^{(n)}\subset D_n , \;  \forall n\in\mathbb{N}\right. \right\}\,,
\]
where $D_n\subset M^n$ is the diagonal of $M^n$, namely the set of points $(x,\dots, x)\in M^n$, and the set of regular functionals
\[
\mathcal{F}_{\text{reg}} := \left\{ F\in \mathcal{F}_{\mu c} \left|  {F^{(n)}(\phi)\in C^\infty_c(M^n),\;\forall n\in\mathbb{N},\;\forall\phi\in\mathcal{C}}\right. \right\}.
\]
The canonical quantization of $(\mathcal{F}_{\mu c},\cdot)$ is realized by deforming the pointwise product.
In the case of a free theory ({\textit{i.e.} when} $\lambda=0$ in \eqref{eq:interacting-equation})
{the star product is defined as}
\[
F\star_\omega G := e^{\hbar\langle \omega,\frac{\delta^2}{\delta \varphi\delta \varphi'} \rangle }\left. F(\varphi) G(\varphi')\right|_{\varphi'=\varphi}\,,
\]
where $\omega\in\mathcal{D}(M^2)'$ is an Hadamard bidistribution, namely, a weak bisolution of the equation of motion up to smooth functions, whose antisymmetric part $\omega(x,y)-\omega(y,x)=i\Delta(x,y)$ is proportional to the causal propagator $\Delta=\Delta_R-\Delta_A$, namely the retarded $\Delta_R$ minus advanced $\Delta_A$ fundamental solutions of the theory. Furthermore, the wave front set of $\omega$ is {\bf microcanonical} \cite{Radzikowski, BFK} so that, at each order in $\hbar$ the $\star_\omega$ product of microcausal functionals is well defined. From now on we shall set $\hbar=1$ if not strictly necessary for the understanding. The set of microcausal functionals equipped with the star product $\mathcal{A}:=(\mathcal{F}_{\mu c},\star_\omega)$ is the $*-$algebra of quantum observables of the free theory. The set of local non linear functionals contains the Wick polynomials of the theory. Algebras constructed with different Hadamard bidistributions are isomorphic, however, single non linear local functionals are not left invariant by this isomorphism.
{For this reason}, the representation of objects like the Wick square or the stress tensor needs to be carefully discussed \cite{HW01,HW05,Moretti}.

In order to implement interactions perturbatively, a time ordering map $T$ needs to be introduced \cite{BF00, EG, HW01, HW02}.
It maps multilocal functionals to microcausal functionals
\[
T\colon{\bigoplus_{n\geq 0}\mathcal{F}_{\textrm{loc}}^{\otimes n}}\to \mathcal{F}_{\mu c}\,,
\]
we refer to \cite{HW01} for its precise definition and for the analysis of the freedom in its construction. The key property for our analysis will be 
the casual factorization property which implies that
\[
T(F,G) = F\star G, \qquad   \; \qquad \text{if}  \qquad F \gtrsim G\,,
\] 
where $F \gtrsim G$ holds if the support of $F$ is not contained in $J^-(\supp G)$. 
Having a time ordering map, a time ordered product can be obtained as $F\cdot_T G = T(T^{-1}(F),T^{-1}(G))$, and hence, $S(F):=\exp_T ({i}F)$, the $S-$matrix of a local functional $F$, can be constructed. Finally the Bogoliubov map  
\[
R_V(F) = S(V)^{-1}\star (S(V) \cdot_T  F )
\]
is used to map the interacting local fields into the algebra of free fields. In general, both $S(F)$ and $R_V(F)$ are given in terms of power series in the coupling constant and thus they are meaningful only in the sense of perturbation theory and when the support of $V$ is compact. 

If the support of $F$ is compact, the causal factorisation property permits to understand the limit of $R_{V_\chi^h}(F)$ where both $\chi$ and $h$ in \eqref{eq:perturbation-lagrangian} tends to $1$ as an inductive limit within $\mathcal{A}$ to all order in perturbation theory. This limit is called {\bf algebraic adiabatic} limit \cite{BF00}. 

The Bogoliubov map is also used to represent the interacting evolution in the free theory. In particular, let $\alpha_t(F)$ be the one parameter group of $*-$automorphisms which represents the free time evolution and which is defined on a local functional $F$ as follows
\[
\alpha_t(F)(\varphi) := F_t(\varphi) := F(\varphi_t)\,,
\]
where, in the fixed Minkowski coordinate system we have chosen, $\varphi_\tau(t,{\bf x}) := \varphi(t-\tau,{\bf x})$. The one parameter group of $*-$automorphisms representing the interacting time evolution is obtained as 
\[
\alpha^V_t(R_V(F)) := R_V(\alpha_t F)\,.
\]
The free and interacting evolution can be interwined by a cocycle $U(t)$ which has been explicitly constructed in \cite{FredenhagenLindner} for observables supported in the region where $\chi$ in $V^h_\chi$ in \eqref{eq:perturbation-lagrangian} is equal to $1$. More precisely
\begin{equation}\label{eq:free-int-evol}
\alpha^V_t(F) = U(t) \star \alpha_t(F) \star U(t)^*\,,
\end{equation}
hence 
\[
U(t+s) = U(t) \star \alpha_t(U(s)).
\]
Furthermore, it satisfies the following equation
\begin{equation}\label{eq:cocycle-K}
\frac{d}{dt} U(t) = {iU(t)\star\alpha_t(K)}\,,
\end{equation}
where the generator $K=K^h=K^{h}_\chi$ of the cocycle is related to $V^h_\chi$ in \eqref{eq:perturbation-lagrangian} as follows  
\begin{equation}\label{eq:gener-coc}
K^h_\chi :=   
R_{V^h_\chi}(\dot{V}^h_\chi), \qquad   \dot{V}^h_\chi  :=   
\int  \dot\chi(t) h({\bf x}) \mathcal{L}_I(t,{\bf x})  dt d^3 {\bf x}\,,
\end{equation}
a direct formula for $\alpha_t^V$ in terms of $K$ can be found in {the appendix, see} \eqref{eq:int-dynamics}.

\subsubsection{KMS states}

Equilibrium states at inverse temperature $\beta$ with respect to the time evolution $\alpha_t$ are assumed to satisfy the Kubo Martin Schwinger (KMS) condition, see \textit{e.g.} \cite{HHW}. We shall now recall the construction of certain KMS states recently presented by Fredenhagen and Lindner \cite{FredenhagenLindner}. 
The state of the free theory we shall start with is $\omega^\beta$, \textit{i.e.} the KMS state which is extremal in the set of KMS states for the free theory. $\omega^\beta$ is a quasi-free state whose two-point function is the following
\begin{equation}\label{eq:2pt-kms-free}
\omega_2^{\beta}(x,y) := \frac{1}{(2\pi)^3} \int dp \;\sigma(p_0) \delta(p^2+m^2)  \frac{1}{1-e^{-\beta p_0}}e^{ip(x-y)}\,,
\end{equation}
where $\delta$ is the Dirac delta function and $\sigma$ is the sign function.
This two-point function is an Hadamard bidistribution and thus, we shall use it to construct the star product of the algebra of the free theory. Hence, from now on $\mathcal{A} =(\mathcal{F}_{\mu c},\star) = (\mathcal{F}_{\mu c},\star_{\omega_2^\beta})$. 

The state for the interacting algebra obtained applying directly $\omega^\beta$ on $R_V(F_1)\star \dots \star R_V(F_n)$ is not a KMS state with respect to $\alpha_t^V$. Nevertheless, Fredenhagen and Lindner \cite{FredenhagenLindner} have shown that using the cocycle $U(t)$ constructed above, the Araki construction of KMS states \cite{Araki-KMS} for perturbed systems can be applied. Hence we may introduce
\begin{equation}\label{eq:int-KMS-state}
\omega^{\beta,V}(A) := \frac{\omega^{\beta}(A\star U(i\beta))}{\omega^{\beta}(U(i\beta))}\,,
\end{equation}
and this is a KMS state for the interacting theory. A direct expansion of $\omega^{\beta,V}(A)$ in terms of free connected $n-$point functions was also given in \cite{Lindner, FredenhagenLindner}. For completeness, this direct representation is recalled in equation \eqref{eq:int-KMS-state-connected} in the appendix. 

The construction recalled above is well-posed, in the sense of perturbation theory, outside the adiabatic limit, namely when $h$ has compact support. 
At the same time, the construction of $U$ depends on the cutoff function $h$ and hence, $\omega^{\beta, V}$ depends on $h$. Hence, the analysis of the well posedness of the state under the adiabatic limit needs to be discussed carefully. 
Adiabatic limits of $\omega^{\beta,V}(A)$ have been analyzed in \cite{FredenhagenLindner} and in \cite{Lindner}. They have shown that,
for the case of a massive quantum field theory, the limit $h\to1$ of $\omega^{\beta,V}(A)$ can be taken in the sense of van Hove. More precisely, according to Definition 2 in \cite{FredenhagenLindner}, a van Hove sequence of cutoff functions is a sequence of compactly supported smooth functions $\{h_n\}_{n\in\mathbb{N}}$ such that 
\begin{align}\label{Equation: definining properties of van Hove sequence}
	0\leq h_{n}\leq 1\,;\qquad
	h_n(x)=1, \quad |x|<n\,; \qquad 
	h_n(x)=0, \quad |x|>n+1 \,.
\end{align}
We say that {a functional $C^\infty_c(\mathbb{R}^3)\ni h\mapsto f(h)\in\mathbb{C}$} converges to $L$ in the sense of van Hove if 
$\lim_{n\to\infty}f(h_n) =L$ for every van Hove sequence and we shall indicate it as
\[
\vlim_{h\to1}f(h)=L\,.
\]
Notice that, for each $V$, the state
$\omega^{\beta,V}$ given in \eqref{eq:int-KMS-state} can be interpreted as a state for the free theory. More precisely, since $U$ is defined as a formal power series in the coupling constant, for every compactly supported $V$ $\omega^{\beta,V}$ is a state for $\mathcal{A}(M)[[\lambda]]$. This property is preserved in the adiabatic limit.
Furthermore, in a similar way, the interacting time evolution $\alpha_t^V$ given in \eqref{eq:free-int-evol} can be seen as $*-$automorphisms of $\mathcal{A}(M)[[\lambda]]$ and 
hence $\omega^{\beta,V}$ are KMS-states with respect to $\alpha_t^V$ at inverse temperature $\beta$ also on $\mathcal{A}(M)[[\lambda]]$.

\subsection{Relative entropy for $C^*-$dynamical systems}

Let us consider a von Neumann algebra $\mathfrak{A}\subset{\mathfrak{B}{\mathcal{H}}}$ and two normal states $\Psi$ and $\Phi$.
The Araki relative entropy \cite{Araki, Araki2, BR, Donald} 
is defined as minus\footnote{We are adopting the definition of Araki, it differs by a minus sing from \cite{BR}. In this way the relative entropy is positive.} the logarithm of the relative modular operator. 
More precisely, the construction starts with the operator $S$, defined as the closure of the operator 
\[
SA\Psi = A^*\Phi , \qquad A\in \mathfrak{A}.
\]
The relative modular operator is then obtained as 
\[
\Delta_{\Psi,\Phi} := S^*S
\]
and the relative entropy is 
\begin{equation}\label{eq:Araki-relative-entropy}
\mathcal{S}(\Psi,\Phi) := -(\Psi, \log(\Delta_{\Psi,\Phi}) \Psi). 
\end{equation}
Unfortunately, this formula cannot be directly applied in the context of pAQFT to test the relative entropy between interacting KMS states \eqref{eq:int-KMS-state} constructed with different interactions Lagrangians because neither the generators of the free or interacting evolutions nor the corresponding modular or relative modular operators are at disposal in $\mathcal{A}$. 
What is available in pAQFT is the relation between the free and interacting evolution given in \eqref{eq:free-int-evol}. Hence, below, we shall generalize the definition of relative entropy specialized to the case of perturbatively constructed KMS states. 
To this end, in appendix \ref{ap:rel-entropy} we derive some expressions for \eqref{eq:Araki-relative-entropy} which involves only the generators of the cocycles intertwining free and interacting evolutions.

\section{Relative entropy in pAQFT}

To discuss thermodynamical properties of states, a notion of entropy is very helpful. In the literature, for the case of $C^*-$dynamical systems whose algebras of observables are von Neumann algebras, also known as $W^*-$dynamical systems, the definition of relative entropy was given by Araki in \cite{Araki}. It plays a key role in the description of thermodynamical relations among states.
Unfortunately, that definition makes use of the relative modular operator and the latter is not at disposal in the context of field theory. 
For this reason, here we generalize that definition in certain cases and we prove that the generalized relative entropy has similar properties as the original extent. 
\begin{definition}\label{def:relative-entropy-KMS-generalized}
Let $V_1$, $V_2$ and $V_3$ be three spatially compact and past compact real perturbation potentials of the form \eqref{eq:perturbation-lagrangian} with fixed $\chi$ and $h$. 
Consider $\omega^{\beta,V_1}$ and $\omega^{\beta,V_3}$ the KMS states obtained by means of perturbation theory from $\omega^\beta$, which is the extremal KMS state  with respect to the time $\alpha_t$ at inverse temperature $\beta$. Consider $\alpha^{V_2}_t$ the one parameter group of $*-$automorphisms obtained perturbing the time evolution $\alpha_t$ with $V_2$.
The {\bf relative entropy} between $\omega^{\beta,V_1}\circ{\alpha_t^{V_2}}$ and $\omega^{\beta,V_3}$ is defined as
\begin{gather}
\mathcal{S}(\omega^{\beta,V_1}\circ{\alpha_t^{V_2}}, \omega^{\beta,V_3}) :=  -\omega^{\beta,V_1}(\beta K_1-\beta K_2) + 
\omega^{\beta,V_1}(\alpha_t^{V_2}(\beta K_3-\beta K_2))
\notag
\\
\label{eq:relative-entropy-KMS-generalized}
-\log (\omega^{\beta}(U_1(i\beta)))+\log (\omega^{\beta}(U_3(i\beta)))\,,
\end{gather}
where $K_i$ are the generators associated to $V_i$ as in \eqref{eq:gener-coc} and $U_i$ the corresponding co-cycles \eqref{eq:cocycle-K}. 
\end{definition}
First of all notice that the expression \eqref{eq:relative-entropy-KMS-generalized} is defined in terms of formal power series in the coupling constant $\lambda$. In particular, $\log(\omega^{\beta,V_1}(U_1(i\beta))) $ can be computed in terms of the connected functions as given in section 4.4.2 of the PhD thesis of Lindner  \cite{Lindner}, see also \eqref{eq:log-cocycle} in the appendix.
As it is clear from the discussion presented in the introduction and in appendix \ref{ap:rel-entropy}, this definition is a generalization of the Araki relative entropy given in \eqref{eq:rel-entropy-vN} to the case of KMS states of a perturbatively constructed quantum field theory. Before analyzing some properties of that expression we notice that in the case where $t$ and either $V_1$ or $V_3$ vanish 
the definition \ref{def:relative-entropy-KMS-generalized} and in particular \eqref{eq:relative-entropy-KMS-generalized} give
\begin{gather*}
\mathcal{S}(\omega^\beta, \omega^{\beta,V}) =    \omega^\beta(\beta K) +\log (\omega^\beta(U(i\beta)))\,,\\
\mathcal{S}(\omega^{\beta,V},\omega^\beta) =  -\omega^{\beta,V}(\beta K) -\log (\omega^\beta(U(i\beta)))\,.
\end{gather*}
\noindent
{\bf Remark} 
Notice that, as discussed in the introduction, $\log (\omega^\beta(U(i\beta)))/\beta$ is nothing but the difference of the free energies in the states $\omega^\beta$ and $\omega^{\beta,V}$. At the same time $K$ is the generator of the co-cycle $U(t)$ which intertwines the time evolutions of $\omega^\beta$ and $\omega^{\beta,V}$. Hence, the previous two expressions recall the definition of entropy as the difference of the internal and free energies multiplied by $\beta$. This is in accordance with the thermostatic formalism introduced in \cite[Section 4.4.]{Lindner}, where, in addition, the first non-trivial order for the free energy is computed, finding an agreement with the results present in the physical literature. This suggests that the present definition should have a direct counterpart in the standard perturbative QFT language. There, thermal equilibrium states are built using the Keldysh contours formalism and perturbative expansions of the propagators through the Matsubara formalism. A direct connection between the two formalisms will be the subject of future investigations.
\\

The generalized relative entropy for perturbed KMS states described so far has similar properties as those shown by the Araki in \cite{Araki} for the case of the relative entropy of states of von Neumann algebras.
Actually the following proposition holds. 

\begin{proposition}\label{pr:rel-entropy-positive}
The generalized relative entropy $\mathcal{S}(\omega^{\beta,V_1}\circ{\alpha_t^{V_2}}, \omega^{\beta,V_3})$ 
 satisfies the following properties:
\begin{itemize}
\item[a)] (Quadratic quantity) the lowest order contribution both in $K_i$ (which are related to $V_i$ as in \eqref{eq:gener-coc}) and in the coupling constant $\lambda$ in 
$\mathcal{S}(\omega^{\beta,V_1}\circ{\alpha_t^{V_2}}, \omega^{\beta,V_3})$ is the second.
\item[b)] (Positivity) $\mathcal{S}(\omega^{\beta,V_1}\circ{\alpha_t^{V_2}}, \omega^{\beta,V_3})$ is positive in the sense of formal power series for every $t$ when $V_1\neq V_3$ or for $t\neq0$ when $V_1=V_3\neq V_2$ and it vanishes in the remaining cases.
\item[c)] (Convexity)   $\mathcal{S}(\omega^{\beta,V_1}\circ{\alpha_t^{V_2}}, \omega^{\beta,V_3})$ is convex in $V_1$, in $V_2$ and also in $V_3$ in the sense of formal power series.
\item[d)] (Continuity) $\mathcal{S}(\omega^{\beta,V_1}\circ{\alpha_t^{V_2}}, \omega^{\beta,V_3})$ is continuous in $V_i$ in the sense of formal power series with respect to the topology of microcausal functionals. 
\end{itemize}
\end{proposition}

\bigskip

\begin{proof}
$a)$ Let us start observing that  
\begin{equation}\label{eq:entr-exp-t}
\mathcal{S}(\omega^{\beta,V_1}\circ{\alpha_t^{V_2}}, \omega^{\beta,V_3}) = \mathcal{S}(\omega^{\beta, V_1},\omega^{\beta, V_3}) 
+ \omega^{\beta,V_1}((\alpha_t^{V_2} - \alpha_t^{V_1})(\beta K_3-\beta K_2)).
\end{equation}
Expanding $\omega^{\beta,V_i}$ and $\omega^{\beta}(\log(U_i(i\beta)))$ in equation \eqref{eq:relative-entropy-KMS-generalized} with \eqref{eq:int-KMS-state-connected} and \eqref{eq:log-cocycle} we obtain the following expansion in powers of $K$:

\begin{gather}
\mathcal{S}(\omega^{\beta,V_1},\omega^{\beta,V_3}) =     
\int_0^\beta du \; \omega^{c,\beta}(\beta K_1 \otimes  \alpha_{iu}K_1)
-\int_0^\beta du \; \omega^{c,\beta}(\beta K_3 \otimes  \alpha_{iu}K_1)   
\notag
\\
\label{eq:second-order}
- \int_{\beta S_2} dU  \omega^{c,\beta}(\alpha_{iu_1}K_1\otimes \alpha_{iu_2}K_1)  
+ \int_{\beta S_2} dU  \omega^{c,\beta}(\alpha_{iu_1}K_3\otimes \alpha_{iu_2}K_3)  
+ O(K^{\otimes 3})\,,
\end{gather}
where we recall that $S_2:=\{(u_1,u_2)\in\mathbb{R}^2|\;0\leq u_1\leq u_2\leq\beta\}$, cfr. equation \eqref{Equation: definition of n-th simplex}, while $\omega^{\beta,c}$ denotes the connected part of the state $\omega^\beta$, see equation \eqref{Equation: definition of connected part}.
Furthermore, in view of \eqref{eq:int-dynamics},
\begin{gather}
\omega^{\beta,V_1}((\alpha_t^{V_2} - \alpha_t^{V_1})(\beta K_3-\beta K_2)) = 
\notag
\\
= - i\beta\int_0^t ds \; \left(  \omega^{c, \beta}(\alpha_s(K_1-K_2)  \otimes  \alpha_{t}(K_3-K_2)))   -  \omega^{c, \beta}(\alpha_t(K_3-K_2)  \otimes  \alpha_{s}(K_1-K_2 ))) \right)
\notag
\\
\label{eq:second-order-bis}
+ O(K^{\otimes 3}).
\end{gather}
Since $K$ is at least of order $1$ in $\lambda$, equations \eqref{eq:entr-exp-t}, \eqref{eq:second-order} and \eqref{eq:second-order-bis} prove $a)$. 

\bigskip
$b)$ In order to prove that $\mathcal{S}(\omega^{\beta,\lambda V_1}\circ \alpha_t^{{\lambda}V_2},\omega^{\beta,\lambda V_3})$ is positive in the sense of formal power series we have to be sure that the lowest order contribution in the coupling constant is positive and that the higher contributions are real. 
Notice that every term in the expansion in powers of $K$ in $\mathcal{S}(\omega^{\beta,\lambda V_1}\circ \alpha_t^{{\lambda}V_2},\omega^{\beta,\lambda V_3})$ is real because $K_i$ is formally selfadjoint for every $i$. 
If we prove that  the second order in $K$ in \eqref{eq:entr-exp-t}, which is obtained from \eqref{eq:second-order} and from \eqref{eq:second-order-bis}, is strictly positive, we prove the sought positivity, because the lowest contribution in the $\lambda$ expansion of the second order expansion in $K$ remains positive. Notice the second order contributions in $K$ are given in \eqref{eq:second-order} and in \eqref{eq:second-order-bis} in terms of connected functions with two entries. 
We thus proceed analyzing the following connected functions for any copies of formally selfadjoint microcausal functionals $A,B$ 
\[
\omega^{c,\beta}(A\otimes B) = \sum_{l} \frac{1}{l!}  \left. D^l_{12} (A\otimes B) \right|_{(0,0)} :=
\sum_{l} \frac{1}{l!}  \left. \langle  A^{(l)}, (\omega_2^\beta)^{\otimes l}B^{(l)} \rangle_{l}\right|_{(\phi,\phi)=(0,0)},
\]
where
$\langle\,,\rangle_l$ denotes the standard pairing between smooth functions over $M^l$, which is tacitly extended to distributions. Moreover,
$\omega_2^\beta$ is the operator obtained by the Schwartz kernel theorem from the two-point function of the free KMS state at temperature $\beta$ given in \eqref{eq:2pt-kms-free}, hence
\[
\omega^{c,\beta}(\alpha_{iu_1}A\otimes \alpha_{iu_2}B) = \sum_{l\geq 1} \frac{1}{l!}  \int dP_l  \left(\prod_{j=1}^{l}  \frac{e^{-w_j(u_2-u_1) }\lambda_+(p_j)+e^{w_j(u_2-u_1 -\beta) }\lambda_-(p_j)}{2w_j (1-e^{-\beta w_j})}\right)  \widehat{\Psi_l}(-P_l,P_l)\,,
\]
where $P_l=(p_1,\dots, p_l)$ with $p_j=(p_{j0},{\bf p}_j)\in\mathbb{R}^4$ and $w_j=\sqrt{{\bf p}_j^2+ m^2}$.
Furthermore, $\lambda_\pm(p_j) = \delta(p_{j0}\mp w_j)$ where $\delta$ is the Dirac delta function and thus $\lambda_\pm(p_j)$ impose the restriction on the positive or negative mass shell of the domain of the $p_j-$integration. 
Finally, $ \widehat{\Psi_l}$ is the Fourier transform of the distribution
\[
\Psi_l(X,Y) = \left. A^{(l)}(X)\otimes B^{(l)}(Y) \right|_{(\phi,\phi)=(0,0)}, \qquad  X,Y\in M^{l}\,.
\]
Notice that, since $A,B$ are formally self-adjoint,  $\widehat{A^{(l)}}(-P)=\overline{\widehat{A^{(l)}}(P)}$. The integrals over every $p_{i0}$ can now be performed thanks to the delta functions supported on the mass shells which are present in $\lambda_\pm$. We obtain
\begin{gather}
\omega^{c,\beta}(\alpha_{iu_1}A\otimes \alpha_{iu_2}B) = 
\notag
\\
\label{eq:simpl-posit}
\sum_{l\geq 1} \frac{1}{l!}  \int d{\bf P}_l 
\prod_{j=1}^{l} 
\left(
\frac{e^{-\frac{\beta}{2} w_{j} } }{2w_{j} (1-e^{-\beta w_{j}})}
\right) 
\sum_{\{E_+,E_-\}} 
\left.
e^{-\sum_{k} p_{k0}\left(u_2-u_1 -\frac{\beta}{2}\right)  }
\widehat{\Psi_l}(-P_l,P_l)\right|_{\substack{ p_{a 0}= \pm w_{a}\\ \forall a\in E_\pm}
}
\end{gather}
where the sum is taken over all possible partitions of $\{1,\dots, l\}$ in two subsets $\{E_+,E_-\} \in P_2\{1,\dots, l\}$ which can also be empty. 

\bigskip

Let us start using equation \eqref{eq:simpl-posit} to expand the second order contributions in \eqref{eq:second-order}. 
Notice that 
the integrals over $u_1$, $u_2$ and $u$ can be taken before the integration over $P$ because $\widehat{\Psi}(-P,P)$ is not of rapid decrease for large momenta only for the directions $P$  for which $f = \sum_{k} p_{k0} = 0$ and ${\bf p}_i=0$ $\forall i$. Furthermore, if $f=0$ and ${\bf p}_i=0$, $\forall i$, $\widehat{\Psi}(-P,P)$ is polynomially bounded in $P$ and its growth is tamed by the factor $e^{-\frac{\beta}{2} \sum_j w_j}$. See Theorem 4 and its proof in \cite{FredenhagenLindner} for further details. 
In particular, using the fact that 
\[
\int_0^{\beta}  e^{-ua + \frac{\beta}{2}a}du   =  2\frac{\sinh \left(\frac{\beta}{2}a\right)}{a} , \qquad 
\int_0^{\beta}du_2\int_{0}^{u_2} du_1   e^{-u_2a + u_1a + \frac{\beta}{2}a} = \beta \frac{\sinh \left(\frac{\beta}{2}a\right)}{a} +R(a), \qquad 
\]  
where $R$ is antisymmetric for changes of $a$ to $-a$, symmetrizing the summand over the partitions $\{E_+,E_-\}$ and noticing that under that symmetrization 
$2 \overline{\widehat{K_i^{(l)}}}{\widehat{K_j^{(l)}}}(P_0,{\bf P})$ is mapped to 
$\overline{\widehat{K_i^{(l)}}}{\widehat{K_j^{(l)}}}(P_0,{\bf P}) + {\widehat{K_i^{(l)}}}\overline{\widehat{K_j^{(l)}}}(P_0,-{\bf P})$
we obtain
\begin{gather}
\mathcal{S}(\omega^{\beta,V_1},\omega^{\beta,V_3}) =   
\sum_{l\geq 1} \frac{1}{l!}  \int d{\bf P}_l 
\prod_{j=1}^{l} 
\left(
\frac{e^{-\frac{\beta}{2} w_{j} } }{2w_{j} (1-e^{-\beta w_{j}})}
\right) 
\sum_{\{E_+,E_-\}} 
\frac{\beta\sinh \left( \frac{\beta}{2}f    \right)}{f }  
\notag
\\
\label{eq:first-contr-positive}
\cdot\left.\left(
\overline{\widehat{K_1^{(l)}}} - \overline{\widehat{K_3^{(l)}} } 
\right)
\left(
{\widehat{K_1^{(l)}}}
-
{\widehat{K_3^{(l)}}}
\right)
\right|_{ p_{a 0}= \pm w_{a},\; a\in E_\pm} + O(K_i^{\otimes 3})\,,
\end{gather}
where 
\begin{equation}\label{eq:sum-frequenceies}
f:= \sum_{k} p_{k0}\,,
\end{equation}
and where the minus sign appearing in front of $-{\bf P}$ is removed by a change of integration variables. 
Notice that since the right hand side of \eqref{eq:first-contr-positive} is a sum of positive quantities
hence the sought positivity is proven for the case $t=0$ and $V_1\neq V_3$.

\bigskip
In order to analyze the remaining cases, in view of \eqref{eq:entr-exp-t}, we need to discuss  
\begin{equation}\label{eq:alternative}
\omega^{\beta,V_1}((\alpha_t^{V_2} - \alpha_t^{V_1})(\beta K_3-\beta K_2))\,.
\end{equation}
Let us start observing that 
\begin{gather}
\omega^{\beta,V_1}((\alpha_t^{V_2} - \alpha_t^{V_1})(\beta K_3-\beta K_2)) = 
- i \beta  \int_0^t ds \; \omega^{\beta} \left( \left[\alpha_s(K_1-K_2), \alpha_t(K_3-K_2)\right]_\star \right) + O(K^{\otimes 3}).
\label{eq:second-order-time}
\end{gather}
Furthermore, from \eqref{eq:simpl-posit} we have that for any $A,B$ formally self-adjoint microcausal functionals 
\begin{gather}
- i \beta  \int_0^t ds \; \omega^{\beta} \left( \left[\alpha_s(A), \alpha_t(B)\right]_\star \right)= 
\beta \sum_{l\geq 1} \frac{1}{l!}  \int d{\bf P}_l 
\prod_{j=1}^{l} 
\left(
\frac{e^{-\frac{\beta}{2} w_{j} } }{2w_{j} (1-e^{-\beta w_{j}})}
\right)
\sum_{\{E_+,E_-\}} 
\sinh\left(\frac{f\beta}{2}\right)
\notag
\\ 
\label{eq:expansion-second-order-time}
\cdot
\left.\left(
\frac{1}{f} (1-\cos(ft)) \left(\overline{\widehat{A^{(l)}}}\widehat{B^{(l)}} + \overline{\widehat{B^{(l)}}}\widehat{A^{(l)}}\right)
-i\frac{\sin(ft)}{f} \left(\overline{\widehat{A^{(l)}}}\widehat{B^{(l)}} - \overline{\widehat{B^{(l)}}}\widehat{A^{(l)}}\right)
\right)
\right|_{\substack{ p_{a 0}= \pm w_{a}\\ \forall a\in E_\pm}},
\end{gather}
where $f$ is given in \eqref{eq:sum-frequenceies} and we have symmetrized the summands over $\{E_+,E_-\}$.
Notice that if both $A = B = K_1-K_2$, the terms proportional to $\overline{\widehat{A}}\widehat{B} - \overline{\widehat{B}}\widehat{A}$ in \eqref{eq:expansion-second-order-time} vanish, while the remaining terms are all formally positive.

\bigskip We now proceed with the discussion of the generic case. If 
\[
A = K_1-K_2, \qquad   B=K_3-K_2
\]
we have that  
\begin{equation}\label{eq:expABsim}
\left(\overline{\widehat{A^{(l)}}}\widehat{B^{(l)}} + \overline{\widehat{B^{(l)}}}\widehat{A^{(l)}}\right) = 2 \left| \frac{\widehat{K_1^{(l)}}+\widehat{K_3^{(l)}}}{2}-\widehat{K_2^{(l)}} \right|^2 - \frac{1}{2} \left|\widehat{K_1^{(l)}}-\widehat{K_3^{(l)}} \right|^2.
\end{equation}
Furthermore, since $0\leq 1-\cos{(ft)}\leq 2$, the negative contributions proportional to $|\widehat{K_1^{(l)}}-\widehat{K_3^{(l)}}|^2$ are controlled by $\mathcal{S}(\omega^{\beta,V_1},\omega^{\beta,V_3})$, as is clear from \eqref{eq:first-contr-positive} and the terms proportional to $\left| \widehat{K_1^{(l)}}+\widehat{K_3^{(l)}}-2 \widehat{K_2^{(l)}} \right|^2$ are formally positive. 
Moreover, 
\begin{equation}\label{eq:expABasim}
\left(\overline{\widehat{A^{(l)}}}\widehat{B^{(l)}} - \overline{\widehat{B^{(l)}}}\widehat{A^{(l)}}\right) = 
\overline{\widehat{K_1^{(l)}}}\widehat{K_3^{(l)}} - \overline{\widehat{K_3^{(l)}}}\widehat{K_1^{(l)}}
- \overline{\widehat{K_2^{(l)}}}\widehat{K_3^{(l)}} + \overline{\widehat{K_3^{(l)}}}\widehat{K_2^{(l)}}
+\overline{\widehat{K_2^{(l)}}}\widehat{K_1^{(l)}} - \overline{\widehat{K_1^{(l)}}}\widehat{K_2^{(l)}}.
\end{equation}
Finally, summing \eqref{eq:first-contr-positive} and \eqref{eq:expansion-second-order-time} composed with \eqref{eq:expABsim} and with \eqref{eq:expABasim} we get
\begin{gather}\label{eq:quadratic-rel-entropy}
\mathcal{S}(\omega^{\beta,V_1}\circ{\alpha_t^{V_2}}, \omega^{\beta,V_3}) = 
\beta \sum_{l\geq 1} \frac{1}{l!}  \int d{\bf P}_l 
\prod_{j=1}^{l} 
\left(
\frac{e^{-\frac{\beta}{2} w_{j} } }{2w_{j} (1-e^{-\beta w_{j}})}
\right)
\sum_{\{E_+,E_-\}} 
\left.
\frac{\sinh\left(\frac{f\beta}{2}\right)}{f}
\widehat{F}\overline{\widehat{F}}
\right|_{\substack{ p_{a 0}= \pm w_{a}\\ \forall a\in E_\pm}}
+O(K^{\otimes 3})\,,
\end{gather}
where 
\begin{equation}\label{eq:def-F}
F =  \sin\left(\frac{ft}{2}\right)\left({K_1+K_3}-2 K_2\right) + i  \cos\left(\frac{ft}{2}\right)\left({K_1-K_3}\right)\,,
\end{equation}
and where $f$ is given in \eqref{eq:sum-frequenceies}. The expression \eqref{eq:quadratic-rel-entropy} implies that the second order contribution of the relative entropy $\mathcal{S}(\omega^{\beta,V_1}\circ{\alpha_t^{V_2}}, \omega^{\beta,V_3})$ cannot be negative.

\bigskip
To conclude the proof of the positivity of the relative entropy we need to prove that the second order contributions are strictly positive unless very special conditions are met.

To reach this goal 
we observe that the generic contribution to \eqref{eq:quadratic-rel-entropy} corresponding to an arbitrary but fixed partition $\{E_+,E_-\}$ is non-negative. Therefore, it is sufficient to analyze in details only one of them, we thus chose the contribution $E_-=\emptyset$ in the sum over partitions $\{E_+,E_-\}$ in the second order term of \eqref{eq:quadratic-rel-entropy}. This gives
\begin{eqnarray*}
C_{E_-=\emptyset} &=& 
\beta \sum_{l\geq 1} \frac{1}{l!}  \int d{\bf P}_l 
\prod_{j=1}^{l} 
\left(
\frac{e^{-\frac{\beta}{2} w_{j} } }{2w_{j} (1-e^{-\beta w_{j}})}
\right)
\left.
\frac{\sinh\left(\frac{f\beta}{2}\right)}{f}
\widehat{F}\overline{\widehat{F}}
\right|_{\substack{ p_{a 0}=  w_{a} \forall a}}
\end{eqnarray*}
where in this case $f=\sum_k w_k$. 
Notice that, if $K_1\neq K_3$, then $\cos(\frac{ft}{2})^2(K_1-K_3)^2$ is positive for fixed $t$ and for almost every $f$. Furthermore, 
 if $K_1= K_3\neq K_2$, $\sin(\frac{ft}{2})^2(K_1+K_3-2K_2)^2$
at fixed $t\neq 0$ is positive for almost every $f$.
In the remaining case $V_1=V_3=V_2$ the relative entropy is trivial because $\omega^{\beta, V_1}\circ\alpha_t^{V_1} =  \omega^{\beta, V_1}$ and the same holds whenever $t=0$ and $V_1=V_3$, {because} $\mathcal{S}(\omega^{\beta,V},\omega^{\beta,V})=0$ for every $\omega^{\beta,V}$.
This concludes the proof of point $b)$.

\bigskip
$c)$ The convexity in $V_i$ for every $i$ can be proved in the sense of perturbation theory analyzing the lowest non vanishing order in $\lambda$ of 
\eqref{eq:quadratic-rel-entropy}. This gives a sum of quadratic elements, namely, all possible $F\overline{F}$ in \eqref{eq:quadratic-rel-entropy} for various 
$l$, ${\bf P}_l$ and $E_\pm$. Since all these elements are convex, we have the thesis.

\bigskip
$d)$ The perturbative expansion of the relative entropy \eqref{eq:quadratic-rel-entropy} guarantees continuity for $V_i$ in $\mathcal{F}_{\mu c}$ with respect to the topology of microcausal functionals in the sense of perturbation theory.
\end{proof}

{\bf Remark}
Notice that while the positivity of both $\mathcal{S}(\omega^{\beta,V_1}\circ\alpha_t^{V_2},\omega^{\beta,V_3})$ and $\mathcal{S}(\omega^{\beta,V_1},\omega^{\beta,V_3})$ is proved in point $b)$ of the previous proposition \ref{pr:rel-entropy-positive}, it is not guaranteed that their difference 
\[
\mathcal{S}(\omega^{\beta,V_1}\circ\alpha_t^{V_2},\omega^{\beta,V_3})- \mathcal{S}(\omega^{\beta,V_1},\omega^{\beta,V_3})=\omega^{\beta,V_1}((\alpha_t^{V_2} - \alpha_t^{V_1})(\beta K_3-\beta K_2))
\]
is positive as can be seen composing \eqref{eq:entr-exp-t} with \eqref{eq:second-order-time} and then with \eqref{eq:expansion-second-order-time}. This will have some implication on the entropy production that we shall introduce and discuss below.

\bigskip

{\bf Remark}
The proof of the positivity of in
$\mathcal{S}(\omega^{\beta,V_1}\circ{\alpha_t^{V_2}}, \omega^{\beta,V_3})$ given in \eqref{eq:entr-exp-t} can be  given in the following alternatively shorter way.
Let us start analyzing the $K$-second order contributions in $\mathcal{S}(\omega^{\beta,V_1},\omega^{\beta,V_3})$ given in \eqref{eq:second-order} which is the first term in $\mathcal{S}(\omega^{\beta,V_1}\circ{\alpha_t^{V_2}}, \omega^{\beta,V_3})$ as displayed \eqref{eq:entr-exp-t}.
Exploiting the KMS condition we notice that the integrals over the simplex $S_2$ can be given in terms of integrals over a single variable. In particular, we obtain that 
\[
\int_{\beta S_2} dU  \omega^{c,\beta}(\alpha_{iu_1}K_1\otimes \alpha_{iu_2}K_1)  =
\frac{\beta}{2}\int_{0}^\infty du \; \omega^{c,\beta}(K_1\otimes \alpha_{iu}K_1).  
\]
We thus obtain 
\begin{equation}\label{eq:first-cont}
\mathcal{S}(\omega^{\beta,V_1},\omega^{\beta,V_3}) = \frac{\beta^2}{2} (K_1-K_3\,|\,K_1-K_3)_{\beta} + O(K^{\otimes 3})
\end{equation}
where we used the Duhamel like two-point function which is a sesquilinear product defined in the following way  
\[
(A\,|\,B)_{\beta}:=\frac{1}{\beta}\int_{0}^{\beta}\omega^{c,\beta} (A^{*}\otimes \alpha_{iu}(B))du,
\]
see \cite[Section 5.3]{BR} for more details and on the properties of this product.
Since, $(\cdot\,|\,\cdot)_{\beta}$ is a positive semidefinite sesquilinear form and $K_1-K_3$ is formally selfadjoint we immediately have the positivity of $\mathcal{S}(\omega^{\beta,V_1},\omega^{\beta,V_3})$.

To treat the second order contributions in the remaining term given in \eqref{eq:second-order-bis} or in \eqref{eq:second-order-time}, we use the analyticity property of the state $\omega^\beta$ and its time translation to transform the time integration into an imaginary time integration thus recognizing again a couple of Duhamel two-point function. Actually we obtain
\begin{equation}\label{eq:second-cont}
\omega^{\beta,V_1}((\alpha_t^{V_2} - \alpha_t^{V_1})(\beta K_3-\beta K_2)) =  \beta^2 (K_1-K_2|K_3-K_2)_\beta - \beta^2 (K_1-K_2|\alpha_t(K_3-K_2))_\beta
+ O(K^{\otimes 3})
\end{equation}
Combining the two contributions \eqref{eq:first-cont} and \eqref{eq:second-cont}, exploiting the sesquilinearity of $(\cdot|\cdot)_\beta$ and its time translation invariance we obtain
\[
\mathcal{S}(\omega^{\beta,V_1}\circ{\alpha_t^{V_2}}, \omega^{\beta,V_3})
= \frac{\beta^{2}}{2} (F|F)_\beta+ O(K^{\otimes 3})
\]
where, in analogy with \eqref{eq:def-F}, $F$ is the formal selfadjoint element \[
F = \frac{1}{2}\left(-\alpha_{t/2}\left({K_1+K_3}-2 K_2\right) +\alpha_{-t/2}\left({K_1+K_3}-2 K_2\right) + \alpha_{t/2}\left({K_1-K_3}\right)+\alpha_{-t/2}\left({K_1-K_3}\right)\right).
\]
Hence, the sesquilinearity of the Duhamel two-point function and the form of $K_i$ implies that $\mathcal{S}(\omega^{\beta,V_1}\circ{\alpha_t^{V_2}}, \omega^{\beta,V_3})$ is positive semidefinite in the sense of formal power series.
To ensure the strict positivity the expansion given in \eqref{eq:quadratic-rel-entropy} can be now directly used.
We thank the anonymous referee for suggesting some steps of this alternative proof.

\section{Adiabatic limits}\label{sec:adiabatic-limit}

In this Section we investigate the adiabatic limit, namely the limit where the spatial supports of $V_i$ of the form \eqref{eq:perturbation-lagrangian} tend to the whole space ($h\to1$).

We shall in particular use some ideas and some technical achievements given in \cite{DHP,FredenhagenLindner} to prove similar results concerning the adiabatic limit of the relative entropy given in  \eqref{eq:relative-entropy-KMS-generalized}.\\
First of all we notice that, in the adiabatic limit, the relative entropy 
 diverges
 due to the integral over an infinite space present in its definition.
In \cite{FredenhagenLindner}, Fredenhagen and Lindner have shown that the adiabatic limit can be taken in the sense of van Hove for the state $\omega^{\beta,V}(A)$ if $A$ is of compact support.
Furthermore, Lindner has shown, in Chapter 4.4 of his PhD Thesis \cite{Lindner}, 
{the finiteness of the van Hove limit v-$\lim_{h\to 1}\log (\omega^\beta(U^h(i\beta)))/ I(h)$} where
\begin{equation}\label{eq:volum-norm}
	I(h):=\int_{\mathbb{R}^3} h({\bf x}) d{\bf x}\,.
\end{equation}
In addition, it is known that another kind of infrared divergences appear when the adiabatic limit is taken in $\omega^{\beta}(\alpha_t^V(A))$, see
 \textit{e.g.} \cite{Altherr, Landsman, LeBellac, SteinmannThermal} and Proposition 4.2 in \cite{DFP} for an explicit computation.

For all these reasons we expect to be able to consider the adiabatic limit of expressions of the form
\[
E=\lim_{h\to 1} \frac{1}{I(h)}\;\mathcal{S}(\omega^{\beta,V_1}\circ{\alpha_t}, \omega^{\beta,V_3}),
\]
which has the dimension of an entropy density.
Hence, we expect to be able to resolve the infinite volume integration discussing the corresponding densities. 
We have actually the following definition

\begin{definition}\label{def:relative-entropy-density}
	Let $V^h_i$ for $i\in\{1,3\}$ be two interaction Lagrangians of the form \eqref{eq:perturbation-lagrangian} with the same spatial cutoff $h\in C^\infty_c(\mathbb{R}^3)$.
	We define the relative entropy per unit volume as
	\begin{equation}\label{eq:def-entropy-unit-volume}
		s(\omega^{\beta,V_1}\circ{\alpha_t}, \omega^{\beta,V_3}) := 
		\vlim_{h\to 1} \frac{1}{I(h)}\;\mathcal{S}(\omega^{\beta,V^h_1}\circ{\alpha_t}, \omega^{\beta,V^h_3})\,,
	\end{equation}
	where $I(h)$ is the integral of the cutoff function over the volume $\mathbb{R}^3$ given in \eqref{eq:volum-norm} and the limit $h\to1$ is taken in the sense of van Hove.
\end{definition}

We now proceed checking that the relative entropy per unit volume is finite, hence the previous definition is well posed. 
We need a preliminary proposition and a remark. Let us start considering the $K^h$ constructed as in \eqref{eq:gener-coc} with an interaction Lagrangian $V$ of the from \eqref{eq:perturbation-lagrangian},
we have that 
\begin{equation}\label{eq:K}
K^h = \int h({\bf x})   {H}({\bf x}) d^3{\bf x},
\end{equation}
for a suitable
$H({\bf x})$, then
the following proposition holds.

\begin{proposition}\label{prop:bounds-l}
Consider the function
\[
l(t,\vettore{x}):=\omega^{\beta,V_1}(\alpha_t(\beta H_3(\vettore{x})))
\]
constructed with $V_1$ and $V_3$ of the from \eqref{eq:perturbation-lagrangian}, with 
\[
H_3(\vettore{x})={R}_{V_3} (\dot{V_3}^{\delta_{\vettore{x}}})
\]
given as in \eqref{eq:K} and where the limit $h\to 1$ has already been taken both in $V_1$ and $V_3$.
The function $l(r,\vettore{x})$ is constant in $\vettore{x}$ and uniformly bounded in $t$.
\end{proposition}
\begin{proof}
To prove this proposition we proceed as in the proof of Theorem 5.1 of \cite{DFP}.
First of all we notice that thanks to Theorem 3 in \cite{FredenhagenLindner} the van Hove limit used in the definition of $l(t,\vettore{x})$ is well defined.
We may thus consider the expansion of $\omega^{\beta,V_1}$ in the adiabatic limit in terms of the connected functions 
\begin{gather}
\omega^{\beta,V_1}(\alpha_t(A)) =  \omega^{\beta}(\alpha_t(A)) 
\notag
\\
\label{eq:mean-ergodic}
+  \sum_{n\geq 1} \int_{\beta S_n} \dvol{u_n} \dots \dvol{u_1} \int_{\mathbb{R}^3} \dvol{^3{\vettore{x}}_1}\dots \int_{\mathbb{R}^3} \dvol{^3{\vettore{x}}_n}  \; \omega^{\beta,c}(\alpha_{t}({A}) \otimes \alpha_{iu_1,{\vettore{x}}_1}(R)\otimes \dots \otimes \alpha_{iu_n,{\vettore{x}}_n}(R))\,,
\end{gather}
where $R := -{R}_{V_1}(\dot{V_1}^{\delta_{0}})$, $A=H_3(\vettore{x})$ and $\alpha_{t,\vettore{x}}$ indicates the automorphisms implementing Minkowski spacetime translations of $(t,\vettore{x})$.
Notice that, thanks to causality and to the form of $\chi$, both $R$ and $A$ have compact support, and as an interacting field $R$ is supported on the points $(t,0)$ where $t\in\supp\dot\chi$. 
Furthermore, as in the proof of Theorem 4 in \cite{FredenhagenLindner}, the connected $n-$point functions can be written by means of the following graphical expansion
\begin{gather*}
\omega^{\beta,c}(\alpha_{iu_0,{\vettore{x}_0}}({A}) \otimes \alpha_{iu_1,{\vettore{x}}_1}(R)\otimes \dots \otimes \alpha_{iu_n,{\vettore{x}}_n}(R))
= \\ 
\sum_{G\in \mathcal{G}^c_{n+1}} \prod_{k<j} \left.\left( \frac{D_{kj}^{l_{kj}}}{l_{kj}!} \right)\cdot
\left(
\alpha_{iu_0,\vettore{x}_0}(A) \otimes \alpha_{iu_1,{\vettore{x}}_1}(R)\otimes \dots \otimes \alpha_{iu_n,{\vettore{x}}_n}(R)
\right)
\right|_{(\phi_0,\dots,\phi_n)=0}\\
=:
\sum_{G\in \mathcal{G}^c_{n+1}} \frac{1}{\text{Symm}(G)} F_{n,G}(u_0,\vettore{x}_0; u_1,\vettore{x}_1;\dots  ;u_n,\vettore{x}_n)\,,
\end{gather*}
where the sum is taken over the oriented connected graphs joining $n+1$ vertices and $\text{Symm}(G)$ is a normalization factor. 
Following the proof of Theorem 4 in \cite{FredenhagenLindner}, $F_{n,G}$ can be computed as
\begin{gather}\label{eq-positive-negative}
F_{n,G}(u_0,{\vettore{x}}_0; u_1,\vettore{x}_1;\dots  ;u_n,\vettore{x}_n) =  \int \dvol{P} \prod_{l\in E(G)}
\frac{e^{i\vettore{p}_l(\vettore{x}_{s(l)}-\vettore{x}_{r(l)})} (\lambda_{+}(p_l)+\lambda_{-}(p_l))}{2w_l(1-e^{-\beta \omega_l})}\cdot
\hat{\Psi}(-P,P)\,,
\end{gather}
where $E(G)$ is the set of lines of the graph $G$, $s(l)$  and $r(l)$ are respectively the indexes of the source and the range of the points joined by the line $l$. Furthermore, $\hat{\Psi}(-P,P)$ is the Fourier transform of $\Psi(X,Y)$ defined as
\[
\Psi(X,Y) = 
\left.
\prod_{l\in E(G)} \frac{\delta^2}{\delta \phi_{s(l)}(x_l)\delta \phi_{r(l)}(y_l)} (A\otimes \underbrace{R\otimes \dots \otimes R}_{n\; \text{times}})
\right|_{(\phi_0,\dots,\phi_n)=0}\,,
\]
so that $X$ and $Y$ stand for $(x_1,\dots ,x_k)$ and $(y_1,\dots ,y_k)$ and $k$ indicates the number of lines in $E(G)$. Hence, in the Fourier transform $\hat{\Psi}(-P,P)$, $P=(p_1,\dots ,p_k)$.

Furthermore, the positive and negative frequency contributions in $D_{ij}$ are then indicated by
\begin{equation}\label{eq:lambdas}
\lambda_+(p_l) = e^{-w_l (u_{r(l)}-u_{s(l)})}\delta(p_{l0}-w_l), \qquad
\lambda_-(p_l) = e^{w_l (u_{r(l)}-u_{s(l)}-\beta)}\delta(p_{l0}+w_l)\,,
\end{equation}
with $w_l=\sqrt{\vettore{p}_l^2+m^2}$. We proceed expanding the products of the sums of positive and negative frequencies parts in \eqref{eq-positive-negative} and performing the 
spatial integrals over $\vettore{x}_1, \dots , \vettore{x}_n$ in $F_{n,G}$. Arguing as in the proof of Theorem 5.1 in \cite{DFP}, we obtain that each graph $G$ in $\mathcal{G}_n^c$ contributes to $\omega^{\beta,V}(\alpha_t(A))$ with a term proportional to

\begin{gather*}
\int_{\beta S_n} \dvol{u_n} \dots \dvol{u_1} \int_{\mathbb{R}^3} \dvol{^3{\vettore{x}}_1}\dots \int_{\mathbb{R}^3} \dvol{^3{\vettore{x}}_n} \;
F_{n,G}(u_0-it,\vettore{x}_0; u_1,\vettore{x}_1;\dots  ;u_n,\vettore{x}_n)\\
=
\left.\sum_{\{E_+,E_-\}\in P_2(E(G))} \int_{\beta S_n} \dvol{U} \int \dvol{\vettore{P}} 
\prod_{l_+\in E_+} \frac{e^{-w_{l_+}(u_{r(l_+)}-u_{s(l_+)})}}{2 w_{l_+}(1-e^{-\beta w_{l_+}})}\cdot
\prod_{l_-\in E_-} \frac{e^{w_{l_-}(u_{r(l_-)}-u_{s(l_-)}-\beta)}}{2w_{l_-}(1-e^{-\beta w_{l_-}})}\cdot
\hat{\Psi}(-P,P) 
\right|_{\substack{p_{k0}=\pm w_{k}, \\ \forall k\in E_\pm}}
\\
\cdot
e^{i\vettore{x}_0 
\left(\sum_{\substack{l\in E(G) \\ s(l)=0}} \vettore{p}_l\right)
}\prod_{j\in \{1,\dots, n\}} \delta\left(\sum_{\substack{l\in E(G) \\ s(l)=j}} \vettore{p}_l-\sum_{\substack{l\in E(G) \\ r(l)=j}} \vettore{p}_l\right)\cdot
\prod_{\substack{e_+\in E_+ \\ s(e_+)=0}} e^{-it w_{e_+}}\cdot\prod_{\substack{e_-\in E_- \\ s(e_-)=0}} e^{it w_{e_-}}\,,
\end{gather*}
where the product of delta functions expresses the momentum conservation.
The exponentials $e^{-it w_{e_+}}$ and $e^{it w_{e_-}} $ are uniformly bounded in time, and the same holds for the results of the remaining  $P$ and $U$ integrations. 
Finally, we observe that the delta functions implementing momentum conservations imply that $\sum_{\substack{l\in E(G) \\ s(l)=0}} \vettore{p}_l{=0}$, and hence  for every graph, $F_{n,G}$ becomes constant in $\vettore{x}_0$ after integration over the other spatial variables, thus concluding the proof.
\end{proof}

\bigskip

{\bf Remark}
We notice that the adiabatic limit, namely the limit $h\to1$ discussed here can be decomposed in two limits. 
Actually, arguing as in Theorem 3 of \cite{FredenhagenLindner}, see also Lemma C.1 in Appendix C of \cite{DHP}, it is possible to observe that the following limits coincide
\begin{equation}\label{eq:vlimit-directions}
\vlim_{h\to 1} \frac{1}{I(h)} \omega^{\beta,V^h}( R_{V^{h}}(\dot{V}^{h})) = 
\vlim_{h_1\to 1} \vlim_{h_2\to 1} \frac{1}{I(h_1)} \omega^{\beta,V^{h_2}}(R_{V^{h_2}}(\dot{V}^{h_1}))\,,
\end{equation}
where $K^h=R_{V^{h}}(\dot{V}^{h})$.
Loosely speaking, the limit $h\to 1$ can be taken in different steps without altering the final result. 

\begin{proposition}\label{prop:rel-ent-pos}
The relative entropy per unit volume $s(\omega^{\beta,V_1}\circ{\alpha_t}, \omega^{\beta,V_3})$ given in \eqref{eq:def-entropy-unit-volume} is finite.
\end{proposition}

\begin{proof}
In \cite[Prop. 4.4.2]{Lindner} it has been shown that the van Hove limit
\begin{align}\label{Equation: Falk's results on densities}
L:=\textrm{v-}\lim_{h\to 1}\frac{\log\big(\omega^\beta(U(i\beta))\big)}{I(h)}
\end{align}
exists and is finite.  Recalling the definition of relative entropy \eqref{eq:relative-entropy-KMS-generalized},
to ensure the finiteness of the relative entropy density \eqref{def:relative-entropy-density} in the adiabatic limit
we just need to prove the finiteness of the following two limits
\begin{align}\label{Equation: densities to be shown to be finite in the adiabatic limit}
L_1 := \vlim_{h\to 1} \frac{1}{I(h)} \omega^{\beta,V^h_1}(\beta K^h_1), \qquad 
L_2 := \vlim_{h\to 1} \frac{1}{I(h)} \omega^{\beta,V^h_1}(\alpha_t(\beta K^h_3)).
\end{align}
We shall now consider only $L_2$ because the same conclusions for $L_1$ follow fixing $t=0$ and $V_3=V_1$. 
Let us decompose $K^h_3$ as in \eqref{eq:K} by introducing ${H}^h_3({\bf x})=R_{V^h_3}({\dot{V}_3}^{\delta_{\vettore{x}}})$ 
\[
K^h_3 = \int h({\bf x})   {H}^h_3({\bf x}) d^3{\bf x}.
\]
Equation \eqref{eq:vlimit-directions} discussed in the remarks above implies that
\[
L_2 := \vlim_{h_1\to 1} \vlim_{h_2\to 1} \frac{1}{I(h_1)}
\int h_1({\bf x})   \omega^{\beta,V_1^{h_2}}(\alpha_t(\beta {H}^{h_2}_3({\bf x}))) d^3{\bf x}.
\]
We taken now the limit $h_2\to1$ and thanks to 
Proposition \ref{prop:bounds-l} we have that 
\[
l_2({\bf x}) := \vlim_{h_2\to 1} \omega^{\beta,V^{h_2}_1}(\alpha_t(\beta H^{h_2}_3({\bf x})))
\]
exists, is constant in ${\bf x}$ and bounded in $t$. We can now take the limit $h_1\to 1$ and from 
\eqref{eq:volum-norm} we have that $L_2=l_2$ thus concluding the proof.
\end{proof}

\begin{proposition}
The relative entropy per unit volume $s(\omega^{\beta,V_1}\circ{\alpha_t}, \omega^{\beta,V_2})$ is positive.
\end{proposition}
\begin{proof}
First of all notice that if $V_1=V_2$ and $t=0$ the relative entropy per unit volume {vanishes} because $\mathcal{S}(\omega^{\beta, V_1},\omega^{\beta, V_1})$ is zero for every $h$. In the other cases consider a van Hove sequence $h_n$ converging to $1$.  
Notice that $I(h_n)$ is positive for every $n$  because  $h_n$ are positive functions. The relative entropy $\mathcal{S}$ is also positive for every $n$ as shown in proposition \ref{pr:rel-entropy-positive}. The limit for $n\to \infty$ of positive quantities is positive, thus we have the thesis.
\end{proof}

\section{Entropy production}

We start this section recalling the definition of entropy production given for $C^*-$dynamical systems in \cite{JP01,JP02, JP03}, see also the previous works \cite{Oj0,Oj1,Oj2} and \cite{Ruelle 2, Ruelle 3} for the case of spin systems. Let $\omega$ be a KMS state with respect to $\alpha_t$ and let $\alpha_t^{V}$ be the dynamics perturbed by $V$, then, the entropy production\footnote{Notice that the sign of $\mathcal{E}_V$, which descends from the sign of $\sigma_V$, differs only apparently from the definition given in \cite{JP01} because in that paper $\beta$ is assumed to be $-1$ as it is common in the context of Tomita-Takesaki theory.}   
   of $\alpha_t^V$ with respect to $\alpha_t$ in the state $\eta$ is usually defined as
\[
\mathcal{E}^V(\eta):=\eta\left( \sigma^{V} \right), \qquad \text{where}\qquad  \sigma^{V}:=  -\left.\frac{d}{dt} \alpha_t(\beta V)\right|_{t=0}
=
-\left.\frac{d}{dt} \alpha^V_t(\beta V)\right|_{t=0}\,,
\]
where the last equality can be obtained from the cocycle condition of $U_V$ and the definition of its generator.  
Unfortunately, in the case of field theories, $\sigma^V$ cannot be easily computed because the explicit form of the Hamiltonian generating $\alpha_t$ is not known. In spite of this fact, a generalization of that formula can be obtained in the case of states which possess a time invariance. 
In particular, if $\eta$ is a state invariant under the one parameter group of weakly continuous automorphisms $\alpha^{V_1}_t$, it holds that 
\begin{gather*}
\mathcal{E}^{V_2}(\eta) = \left.\frac{d}{dt}    \eta(\alpha_{-t}^{V_1}\alpha_t(-\beta V_2))    \right|_{t=0}, 
\qquad
\mathcal{E}^{V_2}(\eta\circ \alpha_s) = \left.\frac{d}{dt}    \eta(\alpha_{-t}^{V_1}\alpha_t(-\beta V_2))    \right|_{t=s}\,,
\\
\mathcal{E}^{V_2}(\eta\circ \alpha^{V_2}_s) = \left.\frac{d}{dt}    \eta(\alpha_{-t}^{V_1}\alpha^{V_2}_t(-\beta V_2))    \right|_{t=s}\,.
\end{gather*}
These expressions can be directly generalized to the case of quantum field theories constructed perturbatively, we actually introduce the following definition valid in that context. 
\begin{definition}\label{def:entropy-production}
Let $V_i$, for $i\in\{1,2,3\}$, be three perturbation Lagrangians of the form \eqref{eq:perturbation-lagrangian} constructed with the same cutoff function $h\in C^{\infty}_0(\mathbb{R}^3)$ which are past compact and of compact spatial support. Consider
 $\eta$, a state which is invariant under the one parameter group of automorphisms $\alpha_t^{V_1}$. 
The entropy production in the state $\eta$ of $\alpha_t^{V_2}$ relative to $\alpha_t^{V_3}$ (or to $\omega^{\beta,V_3}$) is defined as
\begin{equation}
\mathcal{E}^{V_2}_{V_3}(\eta) := 
\left. \frac{d}{dt}\eta\left(\alpha_{-t}^{V_1} \alpha_t^{V_2} (\beta(K_3-K_2))\right)   \right|_{t=0}\,.
\end{equation}
Analogously, the 
entropy production in the state $\eta\circ{\alpha_t^{V_2}}$ of $\alpha_t^{V_2}$ relative to $\alpha_t^{V_3}$
is defined as
\begin{equation}\label{eq:entropy-production}
\mathcal{E}^{V_2}_{V_3}(\eta\circ\alpha_s^{V_2}) := 
\left. \frac{d}{dt}\eta\left(\alpha_{-t}^{V_1} \alpha_t^{V_2} (\beta(K_3-K_2))\right)   \right|_{t=s}\,.
\end{equation}
\end{definition}
In analogy to Theorem 1.1 in \cite{JP01} for the case of $C^*-$dynamical systems, the following proposition, which motivates the name entropy production, holds true:

\begin{proposition}\label{pr:entr-prod-rel-ent}
Consider $V_i$ for $i\in\{1,2,3\}$ three perturbation potentials which are past compact and with spatially compact supports and the KMS state $\omega^{\beta,V_3}$ then
\begin{equation}\label{eq:entr-prod-rel-ent}
\mathcal{S}(\omega^{\beta,V_1}\circ\alpha^{V_2}_t,\omega^{\beta,V_3}) =\mathcal{S}(\omega^{\beta,V_1},\omega^{\beta,V_3})+ \int_{0}^t  \mathcal{E}^{V_2}_{V_3}(\omega^{\beta,V_1}\circ\alpha^{V_2}_s) \; ds
\end{equation}
where $\mathcal{E}^{V_2}_{V_3}(\omega^{\beta,V_1}\circ\alpha^{V_2}_s)$ is the entropy production relative to the KMS state $\omega^{\beta,V_3}$.
\end{proposition}
\begin{proof}
Equation \eqref{eq:entr-exp-t}, the invariance of $\omega^{\beta,V_1}$ with respect to $\alpha_t^{V_1}$ and the fact that $\alpha_0^{V_2}$ is the identity  imply that 
\begin{gather*}
\mathcal{S}(\omega^{\beta,V_1}\circ\alpha^{V_2}_t,\omega^{\beta,V_3}) -\mathcal{S}(\omega^{\beta,V_1},\omega^{\beta,V_3})
=
\omega^{\beta,V_1}((\alpha_t^{V_2} - \alpha_t^{V_1})(\beta K_3-\beta K_2))
=\\
\omega^{\beta,V_1}(\alpha_{-t}^{V_1}\alpha_t^{V_2}(\beta K_3-\beta K_2))
-
\omega^{\beta,V_1}(\alpha_{0}^{V_1}\alpha_0^{V_2}(\beta K_3-\beta K_2))
=
\int_0^{t} \frac{d}{ds} 
\omega^{\beta,V_1}(\alpha_{-s}^{V_1}\alpha_s^{V_2}(\beta K_3-\beta K_2))\,ds.
\end{gather*}
The proof can thus be concluded recalling the definition \ref{def:entropy-production}.
\end{proof}

\bigskip
{\bf Remark} Notice that, the entropy production is not always positive, because as discussed in the remark after proposition \ref{pr:rel-entropy-positive},
the difference $\mathcal{S}(\omega^{\beta,V_1}\circ\alpha^{V_2}_t,\omega^{\beta,V_3}) -\mathcal{S}(\omega^{\beta,V_1},\omega^{\beta,V_3})$ is not necessarily positive. However, if ergodic means (infinite time average) are considered, we have that the entropy production is positive. We shall prove this fact below.

\bigskip
Notice also that from the definition \ref{def:entropy-production} it holds that $\mathcal{E}_{V_3}^{V_1}(\omega^{\beta,V_1}\circ\alpha^{V_2}_t)$ vanishes if $V_2=V_3$. Furthermore, proposition \ref{pr:entr-prod-rel-ent} implies that the entropy production vanishes also in the case $V_1=V_2$.

\bigskip
We shall now rewrite the entropy production in a way which shall be useful in the analysis of time averages.
\begin{proposition}\label{prop:entropy-production-commutator}
In the case of $V_2=0$ it holds that
\[
\mathcal{E}_{V_3}(\omega^{\beta,V_1}\circ\alpha_t) = \beta\omega^{\beta,V_1}(\alpha_t(\Phi_t))
\]
where
\[
\Phi_t = -i  [\alpha_{-t}K_1,K_3]_\star\,.
\]
\end{proposition}
\begin{proof}
Notice that 
\[
\alpha_t^V(A) = U_V(t)\star\alpha_t(A) \star U_V(t)^*\,,
\]
furthermore, the cocycle condition and the form of the generator $K$ imply 
\[
-i\frac{d}{dt}U_V(t)= U_V(t)\star \alpha_t(K)\,.
\]
Starting from the definition of entropy production \eqref{eq:entropy-production} we have 
\begin{align*}
\mathcal{E}_{V_3}(\omega^{\beta,V_1}\circ\alpha_t) & =  \beta \; \omega^{\beta, V_1}\left(\frac{d}{dt}    \alpha_{-t}^{V_1} \alpha_t (K_3)\right) \\
& =  \beta \; \omega^{\beta, V_1}\left(\frac{d}{dt}    \left(U_{V_1}(-t)\star K_3\star U_{V_1}(-t)^{*}  \right) \right) \\
& =  -i \beta \; \omega^{\beta, V_1}\left( U_{V_1}(-t)\star\alpha_{-t }([ K_1, \alpha_t(K_3)]_\star)\star U_{V_1}(-t)^{*} \right) \\
& =  -i \beta \; \omega^{\beta, V_1}\left(\alpha^{V_1}_{-t }\alpha_t ([ \alpha_{-t}K_1, K_3]_\star)\right) \\
& =  -i \beta \; \omega^{\beta, V_1} \circ\alpha_t \left([ \alpha_{-t}K_1, K_3]_\star\right)\,.
\end{align*}
Hence, the thesis holds.
\end{proof}
Notice that similar results have been obtained in another context by Haag and Trych-Pohlmeyer in \cite{Haag Trych}.
We proceed discussing the adiabatic limits of the entropy production and its relations to the relative entropy per unit volume introduced in definition \ref{def:relative-entropy-density}. In particular, we analyze the interplay  with the return to equilibrium property discussed for the case of scalar field theories in \cite{DFP}. In that paper it was proven that return to equilibrium holds if $V$ 
is of spatially compact support.
It is furthermore shown that for a perturbation potential which is not of spatial compact support, $\omega\circ \alpha_t^{V}$ and $\omega^{\beta,V}\circ \alpha_t$ do not converge for large time $t$. For this reason their 
ergodic means  (temporal averages) have been considered.
Furthermore, in \cite{DFP}, it was shown that the limit 
\begin{equation}\label{eq:NESS}
\omega^{+}_V := \lim_{t\to\infty} \vlim_{h\to 1}  \frac{1}{t} \int_0^t ds\;  \omega^{\beta,V^h}\circ \alpha_s
\end{equation}
results in a state (constructed perturbatively) for the free algebra. 
Although not being a KMS with respect to $\alpha_t$, this state is {invariant} under time translations, hence, it is a non equilibrium steady state (NESS) \cite{Ru00}.

We would like to estimate how far is $\omega^+_V$ from the equilibrium state $\omega^\beta$. This could be done estimating their relative entropy, but, 
unfortunately, \eqref{eq:relative-entropy-KMS-generalized} cannot be directly applied. It is slightly easier to analyze the entropy production in $\omega^+_V$ adopting the definition given for example in \cite{JP01}.

Furthermore, to avoid infrared problems, the entropy production per unit volume of a NESS $\omega^{+}$ is then obtained extending Definition \ref{def:entropy-production} to the infinite volume case as previously done in Definition \ref{def:relative-entropy-density} of Section \ref{sec:adiabatic-limit} for the relative entropy. 
Hence,  in close analogy to equation (1.2) of \cite{JP01}, the entropy production per unit volume of $\alpha_t$ in the state $\omega^+_{V_1}$
is defined as
\begin{equation}\label{eq:ent-prod-ness}
e_{V_3}(\omega^+_{V_1}) := \lim_{t\to \infty} \vlim_{h\to 1} \frac{1}{t} \frac{1}{I(h)} \int_0^t ds \;  \omega^{\beta,V^h_1}\circ \alpha_s(\beta \Phi_s), \qquad \Phi_s = -i  [\alpha_{-s}K^h_1,K^h_3]_\star,
\end{equation}
where the normalization factor is given in \eqref{eq:volum-norm} and where we have the reformulation of $\mathcal{E}(\omega^{\beta,V_1}\circ\alpha_t)$ given in Proposition \ref{prop:entropy-production-commutator}.

From the positivity of the relative entropy given in Proposition \ref{pr:rel-entropy-positive}, it descends that if $V_1=V_3=V$, the entropy production per unit volume $e_{V}(\omega^+_{V})$ is positive, actually
\begin{eqnarray*}
e_{V}(\omega^+_{V}) &=& \lim_{t\to \infty} \vlim_{h\to 1} \frac{1}{t} \frac{1}{I(h)} \int_0^{t} ds\;  \omega^{\beta,V^h}\circ \alpha_s(\beta \Phi_s) =
\lim_{t\to \infty} \vlim_{h\to 1} \frac{1}{t} \frac{1}{I(h)} \int_{0}^t  \mathcal{E}(\omega^{\beta,V^h}\circ\alpha_s) \; ds \\
&=&
\lim_{t\to \infty} \vlim_{h\to 1} \frac{1}{t} \frac{1}{I(h)}   \left( \mathcal{S}(\omega^{\beta, V^h}\circ \alpha_t,\omega^{\beta, V^h})- \mathcal{S}(\omega^{\beta, V^h},\omega^{\beta, V^h})\right)
\\&=&
\lim_{t\to \infty} \vlim_{h\to 1} \frac{1}{t} \frac{1}{I(h)}    \mathcal{S}(\omega^{\beta, V^h}\circ \alpha_t,\omega^{\beta, V^h})\,,
\end{eqnarray*}
where in the last but one equality we used the fact that $\mathcal{S}(\omega^{\beta,V},\omega^{\beta,V})=0$. Furthermore, 
the right hand side of the previous equation is positive because, thanks to item $b)$ in Proposition \ref{pr:rel-entropy-positive},  
$\mathcal{S}(\omega^{\beta, V}\circ \alpha_t,\omega^{\beta, V})$ is positive for every $h$ and $I(h)$ is also positive. 
However, we shall now see in the next proposition that, since $\mathcal{S}/I$ is bounded uniformly in $h$ and $t$, the entropy production per unit volume vanishes also in the generic case. 

\begin{theorem}\label{th:vanishing-ep}
Let $V_1, V_3$ be two interaction Lagrangians of the form \eqref{eq:perturbation-lagrangian} constructed with the same cutoff function $h$, the NESS $\omega^+_{V}$ obtained with the time average as described in \eqref{eq:NESS} is thermodynamically simple, namely 
\[
e_{V_3}(\omega^+_{V_1}) = 0\,.
\]
{In other words, the} entropy production per unit volume referred to any interacting KMS state vanishes. 
\end{theorem}
\begin{proof}
Start from \eqref{eq:ent-prod-ness}, then proposition \ref{pr:entr-prod-rel-ent} implies that 
\[
e_{V_3}(\omega^+_{V_1}) = 
\lim_{t\to \infty} \vlim_{h\to 1} \frac{1}{t} \frac{1}{I(h)}   \left( \mathcal{S}(\omega^{\beta, V^h_1}\circ \alpha_t,\omega^{\beta, V^h_3})- \mathcal{S}(\omega^{\beta, V^h_1},\omega^{\beta, V^h_3})\right)\,,
\]
hence \eqref{eq:entr-exp-t} implies that  
\[
e_{V_3}(\omega^+_{V_1}) = 
\lim_{t\to \infty} \vlim_{h\to 1} \frac{1}{t} \frac{1}{I(h)} 
\left( \omega^{\beta,V^h_1}(\alpha_t(\beta K^h_3))  - \omega^{\beta,V^h_1}(\beta K^h_3)\right)\,.
\]
In order to analyze the limits $h\to 1$ and $t\to\infty$ we notice that 
\[
\omega^{\beta,V_1}(\alpha_t(\beta K_3))   =  \int_{\mathbb{R}^3}  h(\vettore{x}_0) \omega^{\beta,V_1}(\alpha_t(\beta H_3(\vettore{x}_0)))   \dvol^3 \vettore{x}_0.
\]

Furthermore, arguing as in the proof of Proposition \ref{prop:rel-ent-pos}, we obtain that the result of
\[
L= \vlim_{h\to1}\frac{1}{I(h)} \omega^{\beta,V_1}(\alpha_t(\beta K_3))\,,
\]
is equal to 
\[
l({t,}\vettore{x}_0) = \vlim_{h\to1} \omega^{\beta,V_1}(\alpha_t(\beta H_3(\vettore{x}_0)))\,
\]
which is a function constant in $\vettore{x}$ and uniformly bounded in time thanks to Proposition \ref{prop:bounds-l}. Hence, it exists a constant $C$ such that
\[
\left| L \right| = |l| \leq C.
\] 
This implies that 
\[
\left|e_{V_3}(\omega^+_{V_1})\right| \leq \frac{C}{t}\,,
\]
for every $t>0$ hence,  $e_{V_3}(\omega^+_{V_1})$ vanishes.
\end{proof}
Theorem \ref{th:vanishing-ep} and Proposition \ref{pr:entr-prod-rel-ent} imply that the non-equilibrium steady state $\omega^+$ defined in \cite{DFP} has vanishing entropy production per unit volume. 
The physical consequence of this fact is that $\omega^+$ is thermodynamically close to $\omega^\beta$, that is $\omega^{+}$ is not far from being an equilibrium state. In the standard Statistical Mechanics, this would imply the NESS to be in the normal folium of the free KMS $\omega^{\beta}$.

\section{Conclusion}

In this paper we have defined the notion of relative entropy between interacting KMS states introduced by Fredenhagen and Lindner \cite{FredenhagenLindner} in the framework of perturbative algebraic quantum field theory. We have shown that this definition is compatible with perturbation theory, in particular, it is quadratic with respect to the coupling constant, it is positive, convex with respect to sums of perturbation Lagrangians and continuous with respect to the topology of microcausal functionals. Furthermore, this definition can be used also for KMS states in the adiabatic limit if the corresponding densities are considered.
In the second part of the paper, we have also analyzed the entropy production for this class of interacting states and, in Proposition \ref{pr:entr-prod-rel-ent} we have proved an analogous result to those presented in \cite{JP01}. Also in this case the adiabatic limit can be taken consistently by passing to densitized quantities. 

The introduction of this formalism allowed us to characterize the non-equilibrium steady state $\omega^{+}$, defined originally in \cite{DFP}, as a thermodynamically simple state, \textit{i.e.} a non-equilibrium state whose entropy production vanishes. Actually, it would be interesting to check if this state could be interpreted as a Generalized Gibbs Ensemble, see \cite{Gogolin} and references therein.

Since entropy is one of the principal tools in the study of non-equilibrium Physics, these two definitions may open the way to a deepened study of the thermodynamics of non-equilibrium steady states in the context of field theory, as to some extent, it is done in the present work. For instance, it would be very interesting to obtain a definition and computation of the energy fluxes which characterize this class of states (see \cite{JP02,JP03}), so getting a formulation of a second law of Thermodynamics in the framework of perturbative Algebraic Quantum Field Theory.

From a conceptual point of view it would be interesting to extend our Definition \ref{def:relative-entropy-KMS-generalized} to a larger class of states.
In this respect, a possibility would be to consider the variational formulations for the relative entropy proposed in \cite{Kosaki}, which may provide a correct generalization to the QFT setting.

Finally, recently, Hollands and Sanders \cite{Hollands Sanders} have introduced a notion of ``relative entanglement measure'' of a normal state as the infimum of the Araki relative entropy among all separable states. It will be interesting to merge those results with the extension of relative entropy analyzed in this paper.

\subsection*{Acknowledgments}
This paper is dedicated to the memory of Claudio Faldino father of one of the authors who recently passed away unexpectedly. In particular, FF wishes to express to him his deepest gratitude for his support, his encouragement and for all the sacrifices he made for no reasons but the love that a father feels for his son.
The authors are grateful to K. Fredenhagen and to V. Jak\v{s}i\'c for helpful discussions on the subject. We also thank L. Bruneau for useful comments on an earlier version of this paper. N.D. was supported by the National Group of Mathematical Physics (GNFM-INdAM).

\appendix

\section{Araki relative entropy for perturbed KMS states}
\label{ap:rel-entropy}

Consider a finite dimensional\footnote{The very same formulas hold in the infinite case too, see \cite{BR}. The finite-dimension assumption allows us to avoid some technicalities which are inessential for the scope of this Appendix.}
 $C^*-$dynamical system formed by the von Neumann algebra $\mathfrak{N}$ and a the automorphisms $\alpha_t$ which implements the time evolution. 
Let $\Omega_0$ be the KMS state at inverse temperature $\beta$ with respect to time translations generated by 
$H$.

Consider three perturbations $P_i$, $i\in\{1,2,3\}$ which are self-adjoint elements of  $\mathfrak{N}$ and KMS states $\Omega_i$ obtained by means of the Araki construction over $\Omega_0$. It holds that 
\[
\Omega_i = \frac{1}{N_i} U_i\Omega_0, \qquad 
U_i = e^{-\frac{\beta}{2} (H+P_i)}e^{\frac{\beta}{2} H},
\qquad N_i^2 = (\Omega_0,U_i^*U_i\Omega_0).
\]
Let $W_i(t)$ be the weakly continuous one-parameter groups of unitary evolutions obtained by means of the Stone theorem from the generators $H+P_i$.

The relative modular operator between the states $\Psi:= W_2(t)\Omega_1$ and $\Phi := \Omega_3$ is obtained starting from 
\begin{align*}
S_{\Psi\Phi} A W_2(t)\Omega_1  &=  A^* \Omega_3=  A^* \frac{1}{N_3} U_3\Omega_0 \\
&= \frac{N_1}{N_3} A^*  U_3 U_1^{-1}\Omega_1 \\
&=  \frac{N_1}{N_3} W_2(t)W_2(t)^* A^*  U_3 U_1^{-1}\Omega_1 \\ 
&= \frac{N_1}{N_3} W_2(t) S_1  {U_1^{-1}}^*  U_3^*  A  W_2(t)  \Omega_1\,,
\end{align*}
where we have used $S_1$, the operator which realizes the conjugation $S_1 A\Omega_1 = A^*\Omega_1$.
Hence
\[
 \Delta_{\Psi\Phi} = 
 \left(\frac{N_1}{N_{3}}\right)^2         U_3 {U_1}^{-1}  S_1^*     S_1 {U_1^{-1}}^*  U_3^*  =   \left(\frac{N_1}{N_{3}}\right)^2  e^{-\beta(H+P_3)} \,,
\]
where we have used the fact that the modular operator of $\Omega_1$ is $\Delta_1=S_1^*S_1=e^{-\beta (H+P_1)}$.  
Hence the relative entropy
\begin{equation}\label{eq:rel-entropy-vN}
\mathcal{S}(\Psi,\Phi) = -\beta (\Omega_1, (P_1-P_2) \Omega_1)  +\beta (W_2(t)\Omega_1, (P_3-P_2) W_2(t)\Omega_1)  - \log (N_1^2) + \log  (N_3^2)\,.
\end{equation}
In particular, one finds
\begin{equation}\label{eq:rel-entropy-vN-bis}
\mathcal{S}(\Omega_1,\Omega_3) = -\beta (\Omega_1, (P_1-P_3) \Omega_1)  - \log (N_1^2) + \log  (N_3^2)\,.
\end{equation}

\section{Further explicit expression for the interacting evolution in pAQFT}\label{Appendix: Further explicit expression for the interacting evolution in pAQFT}
In this appendix, we shall collect some explicit expressions for the interacting evolution, for its cocycle and the corresponding KMS states in pAQFT.

\subsection{ Expansion of the the interaction $\alpha_t^V$}
Once the generator $K^h$ is identified as in \eqref{eq:gener-coc}, the interaction dynamics can be expanded in terms of the free one in the following way.
\begin{equation}\label{eq:int-dynamics}
\alpha_t^{V^h}(A)=\alpha_t(A)+\sum_{n\geq 1}i^n\int_{tS_n}
\left[\alpha_{t_1}(K^h),\dots,\left[\alpha_{t_n}(K^h),\alpha_t(A)
\right]_\star \right]_\star dt_1\ldots dt_n,
\end{equation}
\subsection{Expansion of $\omega^{\beta,V}$ in terms of connected functions}
Following \cite{Lindner, FredenhagenLindner}, see also \cite{Araki-KMS}, if $\omega^\beta$ is the extremal KMS state with respect of the evolution $\alpha_t$ at inverse temperature $\beta$ of the free theory we have that, for every $V$ which is of compact support,
\begin{equation}\label{eq:int-KMS-state-connected}
\omega^{\beta,V}(A)=\omega^\beta(A) + \sum_{n= 1}^\infty(-1)^n\int_{\beta S_n}dU\,\omega^{c,\beta}\left(A\otimes\bigotimes_{k=1}^n\alpha_{iu_k}(K)\right)\,.
\end{equation}
where the integrals are taken over
\begin{align}\label{Equation: definition of n-th simplex}
	S_n{=\{(u_1,\ldots,u_n)\in\mathbb{R}^n|\; 0\leq u_1\leq\ldots\leq u_n\leq 1\}}\,,
\end{align}
the $n-$dimensional simplex of edge $1$ and $K$ is the generator of the cocycle $U$ which intertwines the free and interacting evolution.
The functional $\omega^{\beta,c}$ is the connected part of the state $\omega^{\beta}$, which is defined by the equation
\begin{align}\label{Equation: definition of connected part}
	\omega^\beta(A_1\star\ldots\star A_n)=\sum_{P\in\mathsf{P}\{1,\ldots,n\}}\prod_{I\in P}\omega^{\beta,c}\bigg(\bigotimes_{\ell\in I}A_\ell\bigg)\,\qquad
	\forall A_1,\ldots,A_n\in\mathcal{A}\,,\forall n\in\mathbb{Z}_+\,,
\end{align}
together with the condition $\omega^{\beta,c}(1_{\mathcal{A}})=0$.
Here $\mathsf{P}\{1,\ldots,n\}$ denotes the set of partition of $\{1,\ldots,n\}$ in non-empty subsets.

\subsection{Normalization factor}
In this appendix we would like to recall the way in which $\log(\omega^{\beta}(U(i\beta)))$ can be obtained following the presentation given in the PhD Thesis of Lindner \cite{Lindner}.
Exploiting the properties of the connected functions $\omega^{c,\beta}$ of $\omega^\beta$ one finds that
\begin{gather*}
\omega^\beta\left(U(i\beta)\right)=
\sum_{n=1}^\infty(-1)^n\int_{\beta S_n}dU\,\omega^\beta\left(\alpha_{iu_1}{K}\star\dots\star\alpha_{iu_n}{K}\right)=\\
\exp\left(\sum_{n=1}^\infty(-1)^n\int_{\beta S_n}dU \,\omega^{c,\beta}\left(\alpha_{iu_1}{K}\otimes\dots\otimes\alpha_{iu_n}{K}\right)\right)\,,
\end{gather*}
where the integrals are taken over $S_n$ the $n-$dimensional simplex of edge $1$.
Furthermore, the first equality is nothing but the expansion of $U(i\beta)$ in terms of its generator $K$. Taking the logarithm we obtain:

\begin{equation}\label{eq:log-cocycle}
\log\left(\omega^\beta(U(i\beta))\right)=
\sum_{n=1}^\infty(-1)^n\int_{\beta S_n}dU\,\omega^{c,\beta}\left(\alpha_{iu_1}{K}\otimes\dots\otimes\alpha_{iu_n}{K}\right)\,.
\end{equation}

\end{document}